\documentclass[twoside]{article}
\usepackage[a4paper]{geometry}
\usepackage[latin1]{inputenc} 
\usepackage[T1]{fontenc} 
\usepackage{RR}
\usepackage{hyperref}
\usepackage{amsmath,amssymb,latexsym,wasysym,dsfont,stmaryrd}
\usepackage{float,ifthen,xspace,relsize}
\usepackage{algorithm2e}
\usepackage{mymacros}
\usepackage{graphicx}
\usepackage{gastex}
\usepackage{tikz}
\usetikzlibrary{arrows,automata,shapes.geometric,decorations.pathmorphing,backgrounds,positioning,fit,petri}
\usepackage{wrapfig}
\usepackage{caption}

\newcommand{\one}{1}
\newcommand{\two}{2}

\newcommand{\Act}{Act}

\renewcommand{\val}{\ell}

\renewcommand{\egdef}{:=}

\newcommand{\pow}{\widehat}

\newcommand{\power}[2]{{\mbox{Power}(#1,#2)}}

\newcommand{\Tdown}{T\!\!\downarrow}
\newcommand{\Tup}{T\!\!\uparrow}
\newcommand{\Last}{Last}
\newcommand{\Set}{S}
\newcommand{\Setaux}[1]{S_{#1}}

\def\FUS{\text{\sc FUS}}
\newcommand{\fus}[1]{\FUS_{#1}}
\newcommand{\itexp}[2]{{\mbox{exp}^{#1}(#2)}}
\newcommand{\Output}{Output}
\newcommand{\eqmodels}{\models}

\RRdate{December 2012}

\RRauthor{
Bastien Maubert\thanks{bastien.maubert@irisa.fr}
  \thanks[sfn]{Team S4, Universit\'e de Rennes 1 / IRISA, Rennes, France}
  \and
Sophie Pinchinat\thanks{sophie.pinchinat@irisa.fr}
\thanksref{sfn}
}
\authorhead{Maubert \& Pinchinat}
\RRtitle{Strat\'egies uniformes}
\RRetitle{Uniform Strategies}
\titlehead{Uniform Strategies}
\RRresume{Nous consid\'erons des ar\`enes de jeux pour lesquelles nous
  \'etudions des propri\'et\'es d'uniformit\'e des strat\'egies. Ces propri\'et\'es
  font intervenir des ensembles de parties, qui \'emergent d'une quelconque
  motivation s\'emantique. Typiquement, nous pouvons repr\'esenter des
  contraintes sur les strat\'egies autoris\'ees, comme par exemple \^etre
  observationnelles. Nous proposons un language formel pour sp\'ecifier
  les propri\'et\'es d'uniformit\'e et d\'emontrons sa pertinence en
  reformulant divers probl\`emes connus de la litt\'erature. Noter que la
  capacit\'e \`a corr\'eler diff\'erentes parties ne peut \^etre obtenue par
  aucune logique du temps arborescent \`a moins de l'\'equiper d'une
  modalit\'e suppl\'ementaire, appel\'ee R dans cette contribution. Nous
  \'etudions aussi une proc\'edure de synth\`ese automatique de strat\'egies
  soumises \`a une propri\'et\'e d'uniformit\'e, qui \'etend strictement les
  r\'esultats existants bas\'es sur la logique temporelle standard. Nous
  proposons une solution g\'en\'erique pour le probl\`eme de la synth\`ese
  dans le cas où les ensembles de parties sont caract\'eris\'es par une
  relation binaire d\'efinissable par un transducteur fini. Cette
  solution donne une proc\'edure non--\'el\'ementaire. 
}

\RRabstract{
  We consider turn-based game arenas for which we investigate
  uniformity properties of strategies. These properties involve
  bundles of plays, that arise from some semantical motive. Typically,
  we can represent constraints on allowed strategies, such as
  being
  observation-based. 
  We propose a formal language to specify uniformity properties and
  demonstrate its relevance by rephrasing various known problems from
  the literature. Note that the ability to correlate different plays
  cannot be achieved by any branching-time logic if not equipped with
  an additional modality, so-called R in this contribution.  We
  also study an automated procedure to synthesize strategies subject
  to a uniformity property, which strictly extends existing results
  based on, say standard temporal logics. We exhibit a generic
  solution for the synthesis problem provided the bundles of plays
  rely on any binary relation definable by a finite state
  transducer. This solution yields a non-elementary procedure.
}

\RRmotcle{th\'eorie des jeux, logique temporelle, logique
  \'epist\'emique, strat\'egies
  uniformes}
\RRkeyword{game theory, temporal logic, epistemic logic, uniform strategies}
 \RRprojet{S4}  

\RCRennes 

\begin{document}
\RRNo{8144}
\makeRR   

\newtheorem{theorem}{Theorem}
\newtheorem{lemma}[theorem]{Lemma}
\newtheorem{property}[theorem]{Property}
\newtheorem{corollary}[theorem]{Corollary}
\newtheorem{proposition}[theorem]{Proposition}
\newtheorem{definition}{Definition}
\newtheorem{remark}{Remark}
\newtheorem{example}{Example}
\newcommand{\theoname}[1]{{\bf (#1)}}
\newcommand{\theocite}[1]{{\bf \cite{#1}}}

\newenvironment{proof}{\trivlist\item[\hskip \labelsep {\bf Proof}\enskip]}%
{\unskip\nobreak\hskip 2em plus 1fil\nobreak%
\fbox{\rule{0ex}{1ex}\hspace{1ex}\rule{0ex}{1ex}}%
\parfillskip=0pt \endtrivlist}

\section{Introduction}


In extensive (finite or infinite duration) games, the arena is
represented as a graph whose vertices denote positions of players and
whose paths denote plays. In this context, a strategy of a player is a
mapping prescribing to this player which next position to select
provided she has to make a choice at this current point of the
play. As mathematical objects, strategies can be seen as infinite trees
those of which are obtained by pruning the infinite unfolding of the
arena according to the selection prescribed by this strategy; outcomes
of a strategy are therefore the branches of the trees.

Strategies of players are not arbitrary in general, since players aim
at achieving some objectives: in classic game theory with
finite-duration plays, the reasonable \emph{rationality assumption}
leads players to play in such a way that they maximize their
pay-off. More recently, (infinite-duration) game models have been
intensively studied for their applications in computer science
\cite{apt2011lectures} and logic \cite{booklncs2500}. First,
infinite-duration games provide a natural abstraction of computing
systems' non-terminating interaction \cite{alur2002alternating} (think of a communication protocol between a
printer and its users, or control systems). Second, infinite-duration
games naturally occur as a tool to handle logical systems for the
specification of non-terminating behaviors, such as for the
propositional $\mu$-calculus \cite{emerson91}, leading to a powerful
theory of automata, logics and infinite games \cite{booklncs2500} and
to the development of algorithms for the automatic verification
(``model-checking'') and synthesis of hardware and software systems.

Additionally, the cross fertilization of multi-agent systems and
distributed systems theories has led to equip logical systems with
additional modalities, such as epistemic ones, to capture uncertainty
\cite{sato1977study,lehmann1984knowledge,fagin1991model,parikh1985distributed,ladner1986logic,halpern1989complexity}
, and more recently, these logical systems have been adapted to game
models in order to reason about knowledge, time and strategies
\cite{van2003cooperation,jamroga2004agents,dima2010model}. The whole
picture then becomes intricate, mainly because time and knowledge are
essentially orthogonal, yielding a complex theoretical universe to
reason about. In order to understand to which extent knowledge and
time are orthogonal, the angle of view where strategies are infinite
trees is helpful: Time is about the \emph{vertical} dimension of the
trees as it relates to the ordering of encountered positions along
plays (branches) and to the branching in the tree. On the contrary,
Knowledge is about the \emph{horizontal} dimension, as it relates
plays carrying, e.g., the same information.

As far as we know, this horizontal dimension, although extensively
studied when interpreted as knowledge or observation \cite{arnold02b,
  van2003cooperation,jamroga2004agents,benthem2005epistemic,pinchinat05b,chatterjee2006algorithms,alur2007model,dima2010model},
has not been addressed in its generality. In this paper, we aim at
providing a unified setting to handle it. We introduce the generic
notion of \emph{uniformity properties} and associated so-called
\emph{uniform strategies} (those satisfying uniformity properties). Some
notions of ``uniform'' strategies have already been used, e.g., in the
setting of strategic logics
\cite{van2001games,benthem2005epistemic,jamroga2004agents} and in the
evaluation game of Dependence Logic \cite{vaananen2007dependence},
which both fall into the general framework we present here.

We have chosen to tell our story in a simple framework where games are
described by two-player turn-based arenas in which all information is
put inside the positions, and not on the edges. However, the entire
theory can be adapted to more sophisticated models, \eg with labels on
edges, multi-players, concurrent games, \ldots Additionally, although
uniformity properties can be described in a set-theoretic framework, we
have chosen to use a logical formalism which can be exploited to
address fundamental automated techniques such as the verification of
uniformity properties and the synthesis of uniform strategies --
arbitrary uniformity properties are in general hopeless for
automation. The formalism we use combines the Linear-time Temporal
Logic $LTL$ \cite{gabbay80} and a new modality $\R$ (for ``for all
related plays''), the semantics of which is given by a binary relation
between plays. Modality $\R$  generalizes the knowledge operator $K$ of
\cite{halpern1989complexity} for the epistemic relations of agents in
Interpreted Systems. The semantic binary relations between plays are
very little constrained: they are not necessarily equivalences, to
capture, \eg plausibility (pre)orders one finds in doxastic logic
\cite{hintikka1962knowledge}, neither are they knowledge-based, to
capture particular strategies in games where epistemic
aspects are irrelevant. Formulas of the logic are interpreted over
outcomes of a strategy. The $\R$ modality allows to universally
quantify over all plays that are in relation with the current
play. Distinguishing between the universal quantification over all
plays in the game and the universal quantification over all the
outcomes in the strategy tree yields two kinds of uniform strategies:
the \emph{fully-uniform strategies} and the \emph{strictly-uniform
  strategies}.

As announced earlier, we illustrate the suitability of our notions by
borrowing many frameworks from the literature: strategies for games
with imperfect information,
games with opacity conditions, the non-interference properties of
computing systems, diagnosability of discrete-event systems (with a
proposal for a formal definition of \emph{prognosis}), and finally the
evaluation game  for Dependence Logic. Proofs of
Section~\ref{sec-literature} are omitted due to lack of space, but they are quite simple. Through these examples we
show that both notions, strict uniformity and full uniformity, are
relevant and incomparable. These examples
also demonstrate that
defining uniformity properties with our formal language is 
convenient and intuitive. There are even more
instances of uniform strategies in the literature, but the numerous
examples we give here are already convincing enough to justify the
relevance of the notion.

Next we turn to the automated synthesis of uniform strategies. For
this purpose, we unsurprisingly restrict to finite arenas and to binary relations between
plays that are finitely representable: we use \emph{finite state
  transducers} \cite{berstel1979transductions}, an adequate device to
characterize a large class of binary relations between sequences of
symbols, hence they can be used to relate sequences of positions, \ie
plays. Incidentally, all binary relations that are involved in the
relevant literature seem to follow this restriction. In this context,
we address the problem of \emph{the existence of a fully-uniform
  strategy}: ``given a finite arena, a finite state transducer
describing a binary relation between plays, and a formula expressing a
uniformity property, does there exist a fully-uniform strategy for Player
1?''. We prove that this problem is decidable by designing an
algorithm. This algorithm involves a non-trivial powerset construction
from the arena and the finite state transducer. This construction
needs being iterated a number of times that matches the maximum number
of nested $\R$ modalities in the formula specifying the uniformity
property. As each powerset construction is computed in exponential
time, the overall procedure is non-elementary.  Regarding the decision
problem for the existence
of a strictly-uniform strategy, its decidability is an open problem.\\

The paper is organized as follows. in
Section~\ref{sec-uniformproperties} we set the mathematical framework:
we introduce the notion of uniform strategies and we present the
formal language to specify uniformity properties. Next, in
Section~\ref{sec-literature}, we expose a significant set of six
instances of uniform strategy problems from the
literature. Section~\ref{sec-transducer} is dedicated to a short reminder
about finite state transducers that are used in the strategy synthesis problem addressed
and solved in Section~\ref{sec-decision}. We finish by a discussion on
the work done and perspectives in Section~\ref{sec-discussion}.

\section{Uniform properties} 
\label{sec-uniformproperties}

In this section we define a very general notion of uniform strategies.
We consider two-player turn-based games that are played on graphs with
vertices labelled with propositions.
These propositions represent the relevant information for the uniformity properties one
wants to state. If the game models a dynamic system interacting with its
environment, relevant information can be the value of some (Boolean)
variables. If it models agents interacting in a network, interesting
information can be the state of the communication channels. In games with imperfect
information, it can be what action has just been played. \\

From now on and for the rest of the paper, we let $\AP$ be an
infinite set of \emph{atomic propositions}.\\

An \emph{arena} is a structure $\ga=(V,E,v_I,\val)$ where $V=V_1\uplus
V_2$ is the set of \emph{positions}, partitioned between
positions of Player 1 ($V_1$) and those of Player 2 ($V_2$),
$E\subseteq (V_1\times V_2) \cup (V_2\times V_1)$ is the set of
\emph{edges}, $v_I\in V$ is the \emph{initial position} and
$\val:V\rightarrow
 \parti{\AP}$ is a \emph{valuation function}, mapping each position to the finite
 set of atomic propositions that hold in this position. 

 For $v\in V$, $\TracesInf(v)\subseteq vV^\omega$ is the set of
 \emph{infinite traces} starting in $v$, \ie the set of infinite paths $v_0v_1v_2\ldots$
  in the game graph $(V,E)$, with $v_0=v$, and similarly $\TracesFin(v)\subseteq vV^*$ is the set of
 \emph{finite traces} starting in $v$, \ie the set of finite paths
 $v_0v_1\ldots v_n$  in $(V,E)$, with $v_0=v$ and $n\geq 0$.
 We let $\TracesInf = \cup_{v\in V}\TracesInf(v)$ and $\TracesFin =
 \cup_{v\in V}\TracesFin(v)$.
 Typical elements of $\TracesInf$ are
 $\pi,\pi'$, and $\lambda,\lambda'$ are typical elements of $\TracesFin$.
 $\PlaysInf$ denotes $\TracesInf(v_I)$ and $\PlaysFin$ denotes
 $\TracesFin(v_I)$, respectively the set of infinite and finite plays
 in the game. We shall write $\rho$ instead of $\lambda$ to
 distinguish finite plays from other finite traces.

For an infinite trace $\pi=v_0v_1\ldots$ and $i,j\in
 \mathbb N$, $\pi[i]:=v_i$, $\pi[i,\infty]:=v_iv_{i+1}\ldots \in
 \TracesInf$ and $\pi[i,j]:=v_i\ldots v_j \in \TracesFin$. We
 will use similar notations for finite traces, and $|.| : \TracesFin
 \cup \TracesInf \rightarrow \setn \cup \{\omega\}$ denotes the length of the trace. If $\lambda \in
 \TracesFin$, we let $last(\lambda):=\lambda[|\lambda|-1]$ be the last position of
 $\lambda$. 

A \emph{strategy} for Player $k$, $k\in\{1,2\}$, is a partial function $\sigma
: \PlaysFin \rightarrow V$   that assigns the next position to choose in
every situation in which it is Player $k$'s turn to play. In other words $\sigma(\rho)$
is defined if $last(\rho) \in V_k$.  Let $\sigma$ be a strategy for
Player $k$. We say that  a play $\pi\in\PlaysInf$ 
is \emph{induced by} $\sigma$ if for all $i\geq 0$ such that
$\pi[i]\in V_1$, $\pi[i+1]=\sigma(\pi[0,i])$, and the \emph{outcome
  of} $\sigma$, noted $\out(\sigma)\subseteq \Plays_\omega$, is the
set of all infinite plays that are induced by $\sigma$. 

We want to express properties of strategies that do not concern only 
single traces but rather  sets of correlated traces. We first give a very
abstract definition.

\begin{definition}
  Let $\mathcal G$ be an arena. A \emph{uniformity property} $U\subseteq \mathcal P(\PlaysInf)$ is a
  set of sets of plays in $\mathcal G$.
\end{definition}

\begin{definition}
Let $\mathcal G$ be an arena and $U$ a uniformity property. A strategy
$\sigma$ is $U$\emph{-uniform} if $\out(\sigma)\in U$.
\end{definition}

This definition gives an idea of the notion we want to capture, but
first this set-theoretic definition is not very intuitive, and moreover it
is so expressive that automatically handling this notion in its
generality is hopeless. 
We therefore restrict the notion of uniform strategy by fixing a formal language to specify
uniformity properties. As demonstrated in the next section, the language is powerful enough
to capture plethora of instances from the literature.



\label{sec-restricted}


The proposed language enables to express properties of the dynamics of plays,
and resembles the Linear Temporal Logic ($LTL$)
\cite{gabbay80}. However, while $LTL$ formulas are evaluated on
individual plays (paths), we want here to express properties on
``bundles'' of plays. To this aim, we equip arenas with a
binary relation between finite plays, and we enrich the logic with a
modality $\R$ that quantifies over related plays: the intended
meaning of $\R\phi$ is that $\phi$ holds
in every related play.\\

For the general presentation of the logic, we do not make yet
assumptions concerning the binary relation over plays, as opposed to
Section~\ref{sec-decision} dedicated to decidability issues.\\

We now give
the syntax and semantics of the language $\lang$.

\subsection{Syntax}
\label{sec-syntax}
   
The syntax of the language $\lang$ is the following :

\[\lang :\;\; \varphi,\psi ::= p\mid \neg \varphi \mid \varphi\wedge \psi  \mid \fullmoon \varphi \mid \varphi\until\psi
  \mid \R \varphi\]

where $p$ is in $\AP$. As usual we will use the following notations :
$true:=p\vee \neg p$, $false:= \neg true$, $ \F\varphi := true \until \varphi$,  $\G\varphi:= \neg \F\neg \varphi$, and
$\varphi \mathcal W \psi := \varphi \until \psi \vee \G \varphi $.
In addition we will use the following notation: $\dR\varphi := \neg \R
\neg \varphi$, and  for a formula $\phi\in\lang$,  $\subf(\phi)$ denotes the set of
all its subformulas.

The syntax of $\lang$ is similar to that of linear temporal logic with
knowledge \cite{halpern1989complexity}.
However, we use $\R$ instead of the usual knowledge operator $K$ to emphasize that though it has a strong epistemic flavour, notably
in various application instances we present here, it need not be interpreted in terms of knowledge in general, but merely
as a way to state properties of bundles of plays. 

\begin{definition}
\label{def-depth}
For a formula $\phi\in\lang$, we define the \emph{$\R$-depth} of $\phi$, denoted $\depth(\phi)$, as the maximum number of nested $\R$ modalities in
$\phi$. For $n\geq 0$, let $\langn{n}=\{\phi\in\lang \mid
\depth(\phi)=n\}$ be the set of formulas of $R$-depth $n$.
\end{definition}

We note $\ltl$ the language $\langn{n}$ as it matches the syntax (and
also the semantics) of the Linear-time Temporal Logic of \cite{gabbay80}.


\subsection{Semantics}
\label{sec-semantics}

To give the semantics of $\lang$ we take an arena $\ga=(V,E,v_I,\val)$ and a  relation $\leadsto\;\subseteq
\PlaysFin\times\PlaysFin$. 
A formula $\phi$ of $\lang$ is evaluated at some point $i\in\mathbb N$ of a trace $\pi\in\PlaysInf$,
within a \emph{universe} $\Pi\subseteq \PlaysInf$.
The semantics is given by induction over formulas. 
 \[\begin{array}{lcl}
  \Pi,\pi,i\eqmodels p & \mbox{ if } & p\in \val(\pi[i]) \\
  \Pi,\pi,i\eqmodels \neg \varphi & \mbox{ if } & \Pi,\pi,i\not\eqmodels \varphi \\
  \Pi,\pi,i\eqmodels \varphi \wedge \psi & \mbox{ if }  & \Pi,\pi,i \eqmodels \varphi \mbox{ and } \Pi,\pi,i\eqmodels \psi \\
  \Pi,\pi,i\eqmodels \fullmoon \varphi & \mbox{ if } & \Pi,\pi,i+1 \eqmodels \varphi \\
  \Pi,\pi,i\eqmodels \varphi \until \psi & \mbox{ if } & \mbox{there is }j\geq i \mbox{ such that }\Pi,\pi,j\eqmodels \psi \mbox{ and for all }i\leq k < j,\; \Pi,\pi,k \eqmodels \varphi \\
  \Pi,\pi,i\eqmodels \R \varphi & \mbox{ if } & \mbox{for all }\pi'\in\Pi,j\in \mathbb N \mbox{ such that }\pi[0,i]\leadsto\pi'[0,j]
  ,\;\Pi,\pi',j\eqmodels \varphi  
\end{array}\]

The $LTL$ part is classic. $\R\varphi$ is true at some point of a trace
if $\varphi$ is true in every related finite trace in the universe.

We will sometimes need to evaluate an $\ltl$-formula $\phi$ in a
position $v$ of
an arena, with the classic semantics that $\phi$ holds in $v$ if it holds in every
trace starting in $v$.

Formally, for an arena $\ga=(V,E,v_I,\val)$, a position $v\in V$ and a
formula $\phi\in\ltl$, we write

\[\ga,v\models \phi \mbox{ if }\pi,0\models\phi \mbox{ for all }\pi\in\TracesInf(v)\] 

Here we can omit $\Pi$ because the formula has no $\R$ modality.
  

\begin{definition}
\label{def-unifprop}
Given an arena $\ga=(V,E,v_I,l)$, a \emph{uniformity property} is a
pair $(\leadsto,\phi)$ where $\leadsto$ is a relation  over
$\PlaysFin$ and $\varphi \in \lang$ is a formula.
\end{definition}

Now we define two notions of uniform strategies, which differ only in
the universe the  $\R$ modality quantifies over: $\out(\sigma)$ or $\PlaysInf$ 
 (with the latter,
related plays not induced by the strategy
also count). As we shall see in the
examples of the next section, making a nuance is worthwhile.  
\begin{definition}
\label{def-unifstrat}
  Let $\ga$ be an arena and $(\leadsto,\varphi)$ be
  a uniformity property. A strategy $\sigma$ for Player 1 is 
\begin{itemize}
\item $(\leadsto,\varphi)$-\emph{strictly-uniform} if
for all $\pi\in \out(\sigma)$, $\out(\sigma),\pi,0\eqmodels \varphi$.
\item $(\leadsto,\varphi)$-\emph{fully-uniform} if
for all $\pi\in \out(\sigma)$, $\PlaysFin,\pi,0\eqmodels \varphi$.
\end{itemize}
\end{definition}

The notion of fully-uniform strategy is in a sense weaker than the
strictly-uniform one. Indeed, the fact that a particular play in an
arena verifies a formula $\phi$ in the fully-uniform semantics, \ie
with $\PlaysInf$ and not $\out(\sigma)$ as a universe, is independent
of the strategy. Hence in the general definition of uniform
strategies, a uniformity property $U$ could be defined to be a set of
plays instead of a set of sets of plays, and a strategy $\sigma$ could
be said to be uniform if $\out(\sigma)\subseteq U$ instead of
$\out(\sigma)\in U$, it would still contain the notion of fully-uniform strategies. Strictly-uniform strategies could no longer fit in
this definition though, as deciding whether a play verifies a formula
in this semantics cannot be done without knowing the strategy. In this
sense, strict uniformity is ``stronger'' than full
uniformity. However, the notion of fully-uniform strategies still
enables to characterize tree languages that are not even
$\mu$-calculus definable: indeed, as observed by
\cite{icalp06refinement}, given a infinite tree, Property (*) that
at every depth $d>0$, there exist two nodes, one of which being
labeled by, say $p$, and one of which not being labeled by $p$, cannot
be $\omega$-regular (a pumpimg lemma argument suffices).  Hence this
property cannot be defined by any $\mu$-calculus formula. However, we
are able to characterize arenas whose tree unfolding has this
property: consider the equivalence relation $=_{length}$ which relates
plays with equal length. One easily sees that the existence of a $(=_{length},\G (\dR p
\wedge \dR \neg p))$-fully-uniform strategy is equivalent to say that the tree
unfolding of the arena has Property (*).\\

Remark also that the relation $\leadsto$  plays no role
in Definition~\ref{def-unifstrat} if $\varphi$ does not contain any
$\R$ operator, hence it is a mere $LTL$ formula. Notice
that in this latter case, some standard $\omega$-regular (winning)
conditions can be expressed over plays. The extension to a more
powerful logic, such as the full propositional $\mu$-calculus, in
order to capture all $\omega$-regular properties is {\it a priori} possible. However, for the examples considered in Section~\ref{sec-literature}
this full power is not needed.

\section{Frameworks from the literature}
\label{sec-literature}
In this section we demonstrate that the notion of uniform
strategy of previous Section~\ref{sec-restricted}  enables to embed various
problems from the literature, and in particular that it subsumes two existing notions
of uniform strategies.


\subsection{Games with imperfect information}
\label{sec-K-based}

Games with imperfect information, in general, are games
in which some of the players do not know exactly what is the current
position of the game. This can occur in real games, \eg
poker since one does not know what cards her opponents have in hands,
but also in situations arising from computer science, like for example
a program that observes or controls a system by means of a sub-part of
its variables, the interface, while other variables remain hidden. 
One important aspect of imperfect-information games is that not every
strategy is ``playable''.
Indeed, a player who has imperfect information cannot follow a strategy in which
 different moves are chosen for situations that are indistinguishable
 to her. This is why strategies are required to choose
moves uniformly over observationally equivalent situations. This kind
of strategies is sometimes called \emph{uniform strategies}
 in the
community of strategic logics (\cite{van2001games,benthem2005epistemic,jamroga2004agents}), or \emph{observation-based strategies}
 in the community of computer-science oriented
game theory (\cite{chatterjee2006algorithms}).

In
games with imperfect information, the player's ability to remember
what happened so far along a play is a key point to achieve a winning
strategy, as opposed to \eg  perfect-information parity
games, where memoryless strategies are sufficient.  
Moreover in an imperfect information configuration, it is
necessary to define what situations are indistinguishable to the
player, and this requires defining how much memory she has of what
occurs in a play.
  It is therefore
relevant under an imperfect information assumption to distinguish the
\emph{perfect recall} setting, as opposed to the \emph{imperfect
  recall} one. In the former, the player remembers the whole history of the
observation she had of a play, no matter how long it is, while in the
latter the player forgets a part of the information. An agent with
imperfect information can either be memoryless, \ie he does not
remember anything and takes his decisions only based on the current position, or have a bounded
memory, but this case can be reduced to the memoryless case by putting
the different possible configurations of his memory in the positions
of the arena.


While games with imperfect information and perfect recall have been
studied intensively \cite{reif84,chatterjee2006algorithms,berwanger2008power}, the case
of imperfect recall has received much less attention since paradoxes
concerning the interpretation of such games were raised
\cite{piccione1997absent}.  Nonetheless, relevant problems may be
modeled with imperfect recall: typically, particular computing
resources have very limited memory and cannot remember arbitrarily
long histories. 


In this subsection, we  show that the notion of ``uniform'' or
``observation-based'' strategy can be easily embedded in our notion of
uniform strategy, and this no matter the assumption made on the memory
of the player.

We first consider two-player imperfect-information games as studied
for example in
\cite{reif84,chatterjee2006algorithms,berwanger2008power}. In these games, Player
$\one$ only partially observes the positions of the game, such that
some positions are indistinguishable to her, while Player $\two$ has
perfect information (the asymmetry is due to the focus being on the
existence of strategies for Player $\one$). Arenas are labelled
directed graphs together with a finite set of \emph{actions} $\Act$, and in each round, if
the position is a node $v$, Player $\one$ chooses an available action $a$, and
Player $\two$ chooses a next position $v'$ reachable from $v$ through
an $a$-labelled edge.

We equivalently define this framework in a manner that fits our
setting by putting Player $\one$'s actions inside the positions. We
have two kinds of positions, of the form $v$ and of the form
$(v,a)$. In a position $v$, when she chooses an action $a$, Player
$\one$ actually moves to position $(v,a)$, then Player $\two$ moves
from $(v,a)$ to some $v'$.  So an imperfect-information game arena is
a structure $\mathcal G_{imp}=(\mathcal G,\sim)$ where $\mathcal
G=(V,E,v_I,\val)$ is a two-player game arena with positions
in $V_1$ of the form $v$ and positions in $V_2$ of the form $(v,a)$.
For a position $(v,a)\in V_2$,  we note $(v,a).act :=a$.
$E\subseteq V_1\times V_2\cup V_2\times V_1$, $vE (v',a)$ implies
$v=v'$, $v_I\in V_1$. We assume that $p_1\in\AP$ and for every
 action $a$ in $\Act$,, $p_a\in\AP$. $p_1$ holds in positions belonging
to Player~1, and $p_a$ holds in positions of Player~2 where the last
action chosen by Player~1 is $a$: $\val(v)=\{p_1\}$ for $v\in V_1$, $\val(v,a)=\{p_a\}$ for
$(v,a)\in V_2$. Finally, $\sim\;\subseteq V_1^2$ is an observational
equivalence relation on positions, that relates indistinguishable
positions for Player $\one$. We define its extension to finite plays as the
least relation $\approx$ such that $\rho \cdot v \approx \rho'\cdot v'$ whenever $\rho
\approx \rho'$ and $v \sim v'$, and $\rho \cdot (v,a) \approx \rho' \cdot (v',a')$
whenever $\rho\approx\rho'$, $v\sim v'$ and $a=a'$.

We add the classic
requirement that the same actions must be available in indistinguishable positions: for all
$v,v'\in V_1$, if $v\sim v'$ then $vE (v,a)$ if, and only if,
$v'E (v',a)$. In other words, if a player has different options, she
can distinguish the positions.


\begin{definition} 
A strategy $\sigma$ for Player $\one$ is \emph{observation-based} if 
for all $\rho, \rho' \in v(V_2V_1)^*$, $\rho\approx\rho'$  implies $\sigma(\rho).act=\sigma(\rho').act$.
\end{definition}

 We define the formula
\[\text{\tt
  SameAct}:=\G(p_1\rightarrow\bigvee\limits_{a\in\Act}\!\!\!\R \fullmoon
p_a)\]
which expresses that whenever it is Player $\one$'s turn to play, there is an action $a$ that is played
in all  equivalent finite play.

\begin{theorem}
A strategy $\sigma$ for Player $\one$ is observation-based if, and only if, it is $(\approx,\text{\tt SameAct})$-strictly-uniform.
\end{theorem} 

Here we have to make use of the notion of strict uniformity, and not the
full uniformity. Indeed, after a finite play $\pi[0,i]$ ending in $V_1$,
we want to enforce that the actions in the
next positions of equivalent plays are the same only in those
plays that are induced by the strategy we consider, and not in every
possible play in the game. This would of course not hold as soon as
several choices are possible in $\pi[i]$.

Also, it is interesting to see that this notion could also be embedded
with a simpler formula
 by using a relation that is not an equivalence.
 Define $\leadsto$ as:  $\pi[0,i]\leadsto \pi'[0,j]$ if $\pi[i]\in V_1$, $j=i+1$ and
 $\pi[0,i]\approx\pi'[0,i]$,
 and define the formula: 
\[\text{\tt
  SameAct'}:=\G \bigvee\limits_{a\in\Act}\!\!\!\R
p_a.\]
Then a strategy $\sigma$ for Player $\one$ is observation-based if, and
only if, it is $(\leadsto,\text{\tt SameAct'})$-strictly-uniform.

In this version, the relation is not reflexive, in particular plays ending in $V_2$ are linked to no
play, making $\R\phi$ trivially true for any $\phi$. This is the
reason why we no longer need  to mark positions of $V_1$ with
the proposition $p_1$ and
test whether we are in $V_1$ before we ask for some $p_a$ to hold in
every reachable position.

Finally, notice that to embed the case of imperfect-recall for example, one
would just have to replace $\approx$ with the appropriate relation.




\subsection{Games with epistemic condition}
\label{sec-opacity}

Uniform strategies enable to express winning conditions that have epistemic features, the relevance of which is exemplified by the 
\emph{games with opacity condition} studied in \cite{maubert2011opacity}. In that
case, $\R$ can represent a players' knowledge, or distributed
knowledge between a group of players, or common knowledge, depending
on the winning condition one wants to define.


Games with opacity condition are based on two-player
imperfect-information arenas with a particular winning condition,
called the \emph{opacity condition}, which relies on the knowledge of
the player with imperfect information. In such games, some positions
are ``secret'' as they reveal a critical information that the
imperfect-information player wants to know (in the epistemic sense).

More formally, assume that a proposition $p_S\in \AP$ represents the secret. Let $\mathcal G_{inf}=(\mathcal G,\sim)$ be an
imperfect-information arena as described in Section~\ref{sec-K-based},
with a distinguished set of positions $S\subseteq V_1$ that denotes
the
secret. 
Let $\mathcal G=(V,E,v_I,\val)$ be the arena with
$\val^{-1}(\{p_S\})=S$ (positions labeled by $p_S$ are exactly positions $v \in S$).  The \emph{opacity winning condition} is as
follows.

The \emph{knowledge} or \emph{information set} of Player
$\one$ after a finite play is the set of positions that she considers
possible according to the observation she has:
let $\rho\in\PlaysFin$ be a finite play with $last(\rho) \in
  V_1$. The \emph{knowledge} or \emph{information set} of Player
  $\one$ after $\rho$ is $I(\rho):=\{last(\rho')\mid \rho'\in\PlaysFin,
  \rho\approx\rho'\}$.

An infinite play is winning for Player $\one$ if there exists a finite prefix $\rho$ of this play
whose information set is contained in $S$, \ie $I(\rho)\subseteq S$, otherwise Player $\two$ wins.
Again, strategies for Player $\one$ are required to be observation-based.
It can easily be shown that:
\begin{theorem}
\label{theo-opacity}
A strategy $\sigma$ for Player $\one$ is winning if, and only if,
$\sigma$ is $(\approx,\F \R p_S)$-strictly-uniform.\\
\noindent A strategy $\sigma$ for Player $\two$ is winning if, and only if, $\sigma$ is $(\approx,\G \neg \R p_S)$-fully-uniform.
\end{theorem}
For the second statement of Theorem~\ref{theo-opacity}, we make use of
the notion of full uniformity because we are interested in the
knowledge of Player~$\one$. We consider that she does not know what
strategy Player~$\two$ is playing, hence she may consider possible
some plays that are observationally equivalent to her but not induced
by this strategy. 
 
On the other hand, for the first statement, we use
strict uniformity but could have used full uniformity instead. Indeed, since the
actions chosen by Player~$\one$ are part of what she observes, she
cannot consider possible a finite play that has not followed her
strategy until the point $i$ considered. In
fact, if
$\pi$ is induced by $\sigma$ and $\pi[0,i]\approx \pi'[0,i]$, then $\pi'[0,i]$ is
also induced by $\sigma$. The future of the play may not follow sigma,
but this is not a problem here. Indeed,  the property that we consider on
equivalent plays, $p_S$, does not depend on the future.

Notice that though we chose to illustrate with an example, any winning condition that could be expressed as a formula
of the linear temporal logic with one knowledge operator would fit in our setting.

\subsection{Non-interference}

\emph{Non-interference}, as introduced by \cite{goguen1982security},
is a property evaluated on labelled transition systems which handle
Boolean variables.
Such systems are tuples $(S,\mathcal I,\mathcal
O,\delta,s_I,\Output)$ where $S$ is the set of states, $\mathcal
I=H\uplus L$ is a set of Boolean input variables partitioned into
\emph{high security variables} $H$ and \emph{low security variables}
$L$, $\mathcal O$ is the set of Boolean output variables, $\delta :
S\times 2^{\mathcal I}\rightarrow S$ is the transition function that
maps each pair of state and input variables valuation\footnote{we
  classically confuse valuations over a set B of Boolean variables
  with elements of $2^B$.} to a next state, $s_I$ is the initial state,
and $\Output:S\rightarrow 2^{\mathcal O}$ is the output function that
represents a mapping of states onto valuations of the Boolean output
variables.  We extend the transition function $\delta$ to
$S\times\mathcal (2^{\mathcal I})^*\rightarrow S$ as expected:
$\delta(s,\epsilon)=s$ and $\delta(s,ua)=\delta(\delta(s,u),a)$.

We define the $L$-equivalence, $\sim_L$ over $(2^{\mathcal I})^*$ by
$u\sim_Lu'$ whenever $u$ and $u'$ have the same length and they
coincide on the values of the low security input variables, \ie for all $1 \leq i \leq
length(u)$, for all $l\in L$, $l \in u(i) \equivaut l \in u'(i)$. 
Given an infinite sequence of inputs $w\in(2^{\mathcal
  I})^\omega$, we abuse notation by writing $\Output(w)$ for the infinite
sequence of output variables valuations encountered in the states along
the execution of the system on input $w$.
A system $(S,\mathcal
I,\mathcal O,\delta,s_I,\Output)$ verifies the \emph{non-interference
  property} if for any two finite sequences of inputs $w, w' \in
(2^{\mathcal I})^*$, $w\approx_L w'$ implies $\Output(w)=\Output(w')$. In
other words, the valuations of high security variables have no
consequence on the observation of the system. \\

A first natural problem is to decide the non-interference property of
a system. A second more general problem is a control problem: we want to decide whether there is
a way of restricting the set of input valuations along the executions,
or equivalently to control the environment, so that the system is
non-interfering. By constraining the applied restriction to be
trivial, the former problem is a particular case of the latter.
We can encode the control problem in our setting.

Let $Sys=(S,\mathcal
I,\mathcal O,\delta,s_I,\Output)$ be an instance of the problem, and
write $\Sigma$ for $2^{\mathcal I}$ with typical elements
$a,b,\ldots$ Without loss of generality, we can assume that $Sys$ is
\emph{complete}: every input valuation yields a transition. We define
a two-player game arena that simulates the system, in which Player
$\one$ fixes the \emph{environment}, \ie a subset of the possible inputs in
the current state, and Player $\two$ chooses a particular input among
those.  More formally, let $\mathcal G_{Sys}=(V,E,v_I,\val)$, with
$V=V_1\uplus V_2$, $V_1=(\Sigma\uplus\{\epsilon\})\times S$ and $V_2=S\times 2^{\Sigma}$. A
position $(a,s)\in V_1$ denotes a situation where the system reaches
state $s$ by an $a$-transition, and $(s,A)\in V_2$ denotes a situation
where in state $s$, the set of possible inputs is $A$. The set of
edges $E$ of the arena is the smallest set such that $(a,s)E
(s,A)$ for all $s \in S$, $a \in \Sigma$ and $A \subseteq \Sigma$, and
$(s,A)E (a,\delta(s,a))$ whenever $s\in S$ and $a \in A$.
The initial position of the arena is $v_I=(\epsilon,s_I)$, and
by assuming that $\{p_o\mid o\in 2^{\mathcal O}\}\subseteq \AP$, we
set $\val(a,s)=\val(s,A)=\{p_{\Output(s)}\}$.

By writing $\iota$ for the canonical projection from $V_1 \cup V_2$
 onto $2^{\mathcal I}$ (that is 
$\iota(\epsilon,s_I)=\iota(s,A)=\epsilon$ and $\iota(a,s)=a$) and by extending $\iota$ to finite plays as expected, we let 
$\rho\equiv_L\rho'$ hold whenever $\iota(\rho)\sim_L\iota(\rho')$.
We now  define the formula 
$$\text{\tt SameOutput} := \G \bigwedge\limits_{p_o\in\AP}(p_o\rightarrow \R
p_o)$$
 which captures the property that the valuations of output variables  along
$\equiv_L$-equivalent executions of the system coincide, and we can establish the
following.

\begin{theorem}
  There is a one-to-one correspondence between  $(\equiv_L,\text{\tt
    SameOutput})$-strictly-uniform strategies of Player $\one$ and the controllers which ensure the non-interference property of the system.

  In particular, the trivial strategy of Player $\one$, where from any
  position $(a,s)$ she chooses to move to $(s,\Sigma)$, is
  $(\equiv_L,\text{\tt SameOutput})$-strictly-uniform if,
  and only if, the system has the non-interference property.
\end{theorem}

Here we have to use the strict uniformity as we only want to consider
the executions of the machine allowed by the control represented by
the strategy.

Notice that in order to make this control problem more realistic, one
would seek a maximal permissive strategy/controller so that
environments as ``large'' as possible are computed, but this is out of
the scope of the paper.



\subsection{Diagnosis and Prognosis}
\label{sec-diagnosis}
Diagnosis has been intensively studied, in particular by the
discrete-event systems community (see for example
\cite{sampath95,yoo2002polynomial,cassandras1999discrete}). Informally,
in this setting, a discrete-event system is \emph{diagnosable} if any
occurrence of a faulty event during an execution is eventually
detected. More formally, \emph{diagnosability} is a property of
discrete-event systems which are structures of the form
$Sys=(S,\Sigma,\Sigma_o,\Delta,s_I,F)$, with $S$ the set of states,
$\Sigma$ the set of events, $\Sigma_0\subseteq \Sigma$ the
\emph{observable} events, $\Delta \subseteq S \times \Sigma \times S$
the transition relation, $s_I$ the initial state and $F\subseteq S$
the faulty states; we assume that once a faulty state is
reached, only faulty states can be reached (the fault is
persistent). We can rephrase this problem in our setting, with a single
player simulating the system. Notice that since there is only one player, a strategy defines a unique infinite play. 

Here we assume that $p_f\in\AP$ represents the fact of being faulty. Let $\mathcal G_{Sys}=(V,E,v_I,\val)$, with
$V_1=\emptyset$, $V_2=(\Sigma\uplus\{\epsilon\})\times S$, $(a,s)E(b,s')$ whenever
$(s,b,s')\in\Delta$, $v_I=(\epsilon,s_I)$,  and
$\ell(a,s)=\{f\}$ if $s\in F$, $\emptyset$ otherwise.  We write
$\rho\equiv_{\Sigma_o}\rho'$ whenever the sequences of observable
events underlying $\rho$ and $\rho'$ are the same (these sequences are
obtained from the sequences of positions in the play: for each
position of the form $(a,s)$, keep its letter $a$ if $a\in\Sigma_o$,
and delete it otherwise).
In this game Player 1 never plays, all she does is look at Player 2
simulate the system. There is only one strategy for her, which is to
do nothing, and all possible plays, representing all possible executions of
$Sys$, are in the outcome of this strategy.

\begin{theorem}
  $Sys$ is diagnosable if, and only if, Player 1 has a $(\equiv_{\Sigma_o},\F p_f
  \rightarrow \F \R p_f)$-fully-uniform strategy in $\mathcal G_{Sys}$.
\end{theorem}

Here again, since the outcome of the only possible strategy for Player
1 is the whole set of plays, we could equivalently choose full or
strict uniformity.

Prognosis is a companion of diagnosis, but focuses on the ability to predict
that a fault will occur. Prognosability-like properties can be defined
in our setting. As an example, we aim at saying that a system is \emph{prognosable} whenever the fact that a fault
occurs in a system is known at least one step in advance. 
We define the following formula, which means that either a fault never occurs, or it occurs but we know
it one step before it does.

$$\text{\tt Prognose} :=  (\neg p_f)\mathcal W (\neg p_f\wedge
  \R\fullmoon p_f))$$

Using the same framework as for diagnosis, we can propose:
\begin{definition}
\label{def-prognosis}
  A system $Sys$ is \emph{prognosable} if there is a
  $(\equiv_{\Sigma_o},\texttt{Prognose})$-fully-uniform strategy for Player 1 in $\mathcal G_{Sys}$.
\end{definition}

\subsection{Dependence Logic}
Dependence Logic is a flourishing topic introduced recently by
V\"a\"an\"anen \cite{vaananen2007dependence}. It extends first order
logic by adding atomic dependence formulas $\dep(t_1,\ldots,t_n)$,
which express functional dependence of the term $t_n$ on the terms
$t_1,\ldots,t_{n-1}$. A dependence atom $\dep(x_0,x_1)$ can be
interpreted as "the value of $x_1$ depends only on the value of
$x_0$", or "the value of $x_1$ is fully determined by the value of
$x_0$".
 Evaluating a dependence between terms on a
single assignment of the free variables is meaningless: in order to tell whether $t$ depends on $t'$, one
must vary the values of $t'$ and see how the values of $t$ are affected. This is why a formula of
Dependence Logic is evaluated on a first-order model $\mathcal M$ and a \emph{set of assignments} for the free variables, called
a \emph{team}. If $t$ is a term, $\mathcal M$ a model and $s$ an assignment for the free variables
in $t$, we note $\llbracket t\rrbracket^{\mathcal M}_s\in M$ the interpretation
of $t$ in the model $\mathcal M$ with the assignment $s$.

Dependence Logic is inspired by Independence Friendly logic (IF logic), 
a logic defined by Hintikka and
Sandu \cite{hintikka1989informational}.  Van Benthem gives in
\cite{benthem2005epistemic} an imperfect-information
evaluation game for IF logic, using a notion of  uniform strategy
that corresponds with  the classic notion of imperfect-information or
observation-based strategy.

For Dependence Logic, an evaluation game  is also given in
\cite{vaananen2007dependence}. It is presented as a game with imperfect
information, because strategies
must verify some ``uniformity constraint'', which makes the game
undetermined.  However,
 the notion of uniform strategy in these games
 is not defined as ``playing uniformly in
positions of an
information set'', but rather as ``playing such that, when positions
of an information set are reached, some property
 uniformly holds in these positions''. For this reason it is a notion of uniform strategy
different from the one used by Van Benthem in
\cite{van2001games,benthem2005epistemic}, and this game is not a game with
imperfect information \emph{stricto sensu}.   

We present the evaluation game, the uniformity requirement and then show that this notion fits in our setting.

Let $\Phi$ be a sentence (formula with no free variable) of Dependence
Logic in negation normal form, \ie only atomic formulas can be negated, and let $\mathcal M$ be a first order model.  $\mathcal
G^{\mathcal M}(\phi)$ is a two player game between Player $\one$ and
Player $\two$; positions are of the form $(\varphi,n,s)$, where
$\varphi$ is a subformula of $\Phi$, $n$ is the position in $\Phi$ of
the first symbol of $\varphi$ and $s$ is an assignment whose domain
contains the free variables of $\varphi$. The index $n$
is used to decide, given two positions containing the same dependence atom, whether
they are the same \emph{syntactic} subformulas of $\phi$ or not. For a subformula $\varphi$,
$len(\varphi)$ is the number of symbols in $\varphi$.
The game starts in position $(\phi,1,\emptyset)$ and the rules are as
follows:
\begin{center}
\begin{tabular}{ll}
position $(t_1=t_2,n,s)$: &  if $\llbracket
  t_1\rrbracket^{\mathcal M}_s=\llbracket t_2\rrbracket^{\mathcal
    M}_s$, Player $1$ wins. \\
position $(\neg (t_1=t_2),n,s)$: &  if $\llbracket
  t_1\rrbracket^{\mathcal M}_s=\llbracket t_2\rrbracket^{\mathcal
    M}_s$, Player $2$ wins \\
position $(Rt_1\ldots t_m,n,s)$: &  if $\llbracket R\rrbracket^{\mathcal M}\llbracket t_1\rrbracket^{\mathcal
    M}_s\ldots \llbracket t_m\rrbracket^{\mathcal
    M}_s$, Player $1$ wins. \\
position $(\neg Rt_1\ldots t_m,n,s)$: &  if $\llbracket R\rrbracket^{\mathcal M}\llbracket t_1\rrbracket^{\mathcal
    M}_s\ldots \llbracket t_m\rrbracket^{\mathcal
    M}_s$, Player $2$ wins. \\
position $(\dep(t_1,\ldots,t_m),n,s)$: &  Player $1$ wins.\\
position $(\neg \dep(t_1,\ldots,t_m),n,s)$: &  Player $2$ wins.\\
position $(\varphi\vee\psi,n,s)$: &  Player $1$ chooses between
position $(\varphi,n,s)$ and $(\psi,n+1+\mbox{len}(\varphi),s)$.\\
position $(\varphi\wedge\psi,n,s)$: &  Player $2$ chooses between position $(\varphi,n,s)$ and $(\psi,n+1+\mbox{len}(\varphi),s)$.\\
position $(\exists x\varphi,n,s)$: &  Player $1$ chooses a value $a\in M$
and moves to $(\varphi,n+2,s(x\mapsto a))$ \\
position $(\forall x\varphi,n,s)$: &  Player $2$ chooses a value $a\in M$
and moves to $(\varphi,n+2,s(x\mapsto a))$ 

\end{tabular}
\end{center}

\begin{figure}[htb]
  \begin{center}
  \resizebox{\textwidth}{!}{
    \begin{tikzpicture}[shorten >= 1pt, shorten <= 1pt]
      \node (e) {$\begin{array}{c}\forall x_0\forall x_1\phi'\\ \emptyset\end{array}$}; 
      \node (c) [below =of e] {$\begin{array}{c}\forall x_1\phi' \\  \{x_0 \mapsto 1\}\end{array}$};
      \node (l) [left = 6cm of c] {$\begin{array}{c}\forall x_1\phi' \\  \{x_0 \mapsto 0\}\end{array}$};
      \node (r) [right = 6cm of c] {$\begin{array}{c}\forall x_1\phi'\\ \{x_0 \mapsto 2\}\end{array}$};

      \node (lc) [below =of l] {$\begin{array}{c}x_0=x_1\vee \dep(x_0,x_1)\\ \{x_0 \mapsto 0, \\ \;\;\, x_1 \mapsto 1\}\end{array}$};
      \node (auxc) [below = 2cm of lc] {};
      \node [left = 2pt of auxc,fill=red!20] (lcl) {$\begin{array}{c}x_0=x_1\\ \{x_0 \mapsto 0, \\ \;\;\, x_1 \mapsto 1\}\end{array}$};
      \node [right = 2pt of auxc,fill=green!20] (lcr) {$\begin{array}{c}\dep(x_0,x_1)\\ \{x_0 \mapsto 0, \\ \;\;\, x_1 \mapsto 1\}\end{array}$};

      \node (ll) [left = 1.4cm of lc] {$\begin{array}{c} x_0=x_1\vee \dep(x_0,x_1)\\\{x_0 \mapsto 0, \\ \;\;\, x_1 \mapsto 0\}\end{array}$};
      \node (auxl) [below = 2cm of ll] {};
      \node [left = 2pt of auxl,fill=green!20] (lll) {$\begin{array}{c} x_0= x_1\\\{x_0 \mapsto 0, \\ \;\;\, x_1 \mapsto 0\}\end{array}$};
      \node [right = 2pt of auxl,fill=green!20] (llr) {$\begin{array}{c}\dep(x_0,x_1)\\\{x_0 \mapsto 0, \\ \;\;\, x_1 \mapsto 0\}\end{array}$};
      
      \node (lr) [right = 1.4cm of lc] {$\begin{array}{c}x_0=x_1\vee \dep(x_0,x_1)\\ \{x_0 \mapsto 0, \\ \;\;\, x_1 \mapsto 2\}\end{array}$};
      \node (auxr) [below = 2cm of lr] {};
      \node [left = 2pt of auxr,fill=red!20] (lrl) {$\begin{array}{c}x_0=x_1\\ \{x_0 \mapsto 0, \\ \;\;\, x_1 \mapsto 2\}\end{array}$};
      \node [right = 2pt of auxr,fill=green!20] (lrr) {$\begin{array}{c}\dep(x_0,x_1)\\ \{x_0 \mapsto 0, \\ \;\;\, x_1 \mapsto 2\}\end{array}$};
      
      \node (dots) [right = 4cm of lr] {$\ldots$};
      
        \path[thick, ->] (e) edge (l) edge (c) edge (r)
        (l) edge (ll) edge (lc) edge (lr);
        
        \path[thick, ->] (ll) edge (llr);
        \path[thick,->] (lc) edge (lcl);
        \path[thick,->] (lr) edge (lrl); 
        
        \path[line width = 2pt, red, ->] (ll) edge (lll);
        \path[line width = 2pt, red, ->] (lr) edge (lrr);
        
        \path[line width = 2pt, red, ->] (lc) edge (lcr);
        
        \path[dashed, blue] (lcr.south) edge [bend right] (lrr.south);
          

\end{tikzpicture}
  }
\end{center}
  \caption{Evaluation game for $\forall x_0\forall x_1 (x_0=x_1\vee
    \dep(x_0,x_1))$ with $M=\{0,1,2\}$}  
  \label{fig-dep}
\end{figure}
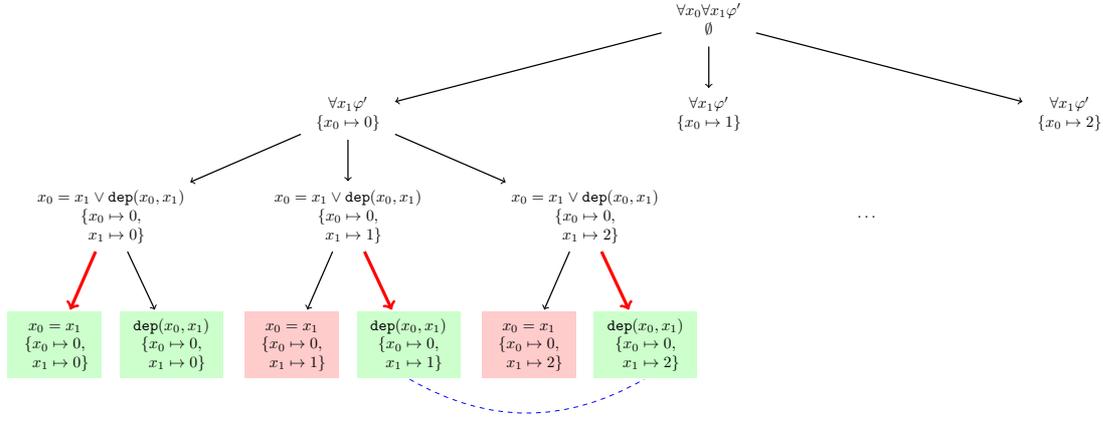

Figure~\ref{fig-dep} represents (a part of) the game for the
evaluation of the Dependence Logic formula $\forall x_0\forall x_1
(x_0=x_1\vee \dep(x_0,x_1))$ on a model $\mathcal M$ with domain
$M=\{0,1,2\}$. This
formula is not true in this model. Indeed, intuitively, if there are
at least three elements in the domain, it is not because $x_0$ and
$x_1$ do not have the same value that the value of $x_0$ determines the
value of $x_1$: there remain two possibilities for the value of $x_1$
that are not the one of $x_0$. So there should not be a winning
strategy for Player 1 for the evaluation game to be correct.

 In the first two rounds, Player~2 chooses a value for
each of the universally quantified variables
$x_0$ and  $x_1$. Then Player~1 chooses a disjunct and we
reach atomic formulas. Green positions are winning for Player~1, red
ones are winning for Player~2. The red arrows indicate a 
strategy for Player~1 (we focus on the subtree for $x_0=0$, we assume that
Player 1 plays the same way in the others). We see that this strategy is winning for Player 1, while
the formula is not true in the model. 

The problem comes from the fact
that, as said earlier, a dependence atom must not be evaluated on a
single assignment but on a \emph{set} of assignments, a team. In the
evaluation game, the team in which a dependence atom $\dep(t_1,\ldots,t_n)$ should be
evaluated is the set of assignments  in leafs that contain
$\dep(t_1,\ldots,t_n)$ and are reached by the strategy. In the
example, the assignments in the two leafs linked with the  dashed
line, $\{x_0\mapsto 0,x_1\mapsto 1\}$ and $\{x_0\mapsto 0, x_1\mapsto 2\}$ , are thus part of the team in which $\dep(x_0,x_1)$ should be
evaluated (there are more in the two other subtrees not shown here).
Then we see that this strategy should not be allowed as while both
leaves agree on the value of $x_0$, they do not agree on the value of $x_1$.
This observation leads to defining a certain notion of uniform strategy.

A strategy $\sigma$ for Player $\one$ is \emph{uniform} in the sense of
\cite{vaananen2007dependence} if, for every two finite plays $\rho,\rho'\in \out(\sigma)$
such that $last(\rho)=(\dep(t_1,\ldots,t_m),n,s,\one)$ and
$last(\rho')=(\dep(t_1,\ldots,t_m),n,s',\one)$ contain the same (syntactically
speaking) atomic dependence subformula, if $s$ and $s'$ agree on
$t_1,\ldots,t_{m-1}$, then they also agree on $t_m$. Then we have the
expected property that a sentence $\phi$ of
Dependence Logic is true in a model $\mathcal M$ if Player $\one$ has
a winning \emph{uniform} strategy in $\mathcal G^{\mathcal M}(\phi)$.

We characterize uniform strategies in the sense of \cite{vaananen2007dependence} as uniform strategies in our sense.
The game described above easily fits in our setting (we add loops on terminal positions so as to obtain infinite plays).
Let $\Phi$ be a sentence of Dependence Logic, and $\mathcal M$ be a
finite model. We call $G_\Phi^{\mathcal M}=(V,E,v_I)$ the evaluation
game defined above.
For each object $a\in M$ of the domain we use one atomic proposition $p_a$,
and we also use the proposition $p_d$ to mark positions that contain dependence atoms. 
So we assume that $\{p_a\mid a\in M\}\uplus\{p_d\}\subset \AP$,
and  we define $\ga_\Phi^{\mathcal M}=(V,E,v_I,\val)$, where the valuation $\val$ is as follows :

\[\begin{array}{rcl}
\val(\dep(t_1,\ldots,t_m),n,s) & = &
\{p_a,p_d\}   \mbox{ with } a=\llbracket t_m\rrbracket^{\mathcal M}_s 
\\
\val(\;\_\;,n,s) & = &\emptyset
\end{array}\]

We define the equivalence relation $\simeq$ on finite plays as the smallest reflexive relation such that if there is
$\varphi=\dep(t_1,\ldots,t_m)$ and $n$ s.t $last(\rho)=(\varphi,n,s)$, $last(\rho')=(\varphi,n,s')$,
 and $\llbracket t_i\rrbracket^{\mathcal
    M}_s=\llbracket t_i\rrbracket^{\mathcal
    M}_{s'}$ for $i=1,\ldots,n-1$, then $\rho \simeq \rho'$. Now we define the formula 
\[\text{\tt AgreeOnLast}: = \G(p_d\rightarrow \bigvee\limits_{a\in M}\R p_a)\] which expresses that
whenever the current position contains a dependence atom $\dep(t_1,\ldots,t_m)$, it agrees with all equivalent
finite plays on some value $a$ for $t_m$. Since equivalent plays are those ending in a position that has the same 
dependence atom and agrees on the first $m-1$ terms, it is easy to prove: 

\begin{theorem}
A strategy $\sigma$ for Player $\one$ in $G_\Phi^{\mathcal M}$ is
uniform if, and only if, it is $(\simeq,\text{\tt
  AgreeOnLast})$-strictly-uniform in $\ga_\Phi^{\mathcal M}$.
\end{theorem}

In this example again, the strict uniformity is needed, as we only
want to catch leaves that are hit by the strategy.
Also, note that here $\simeq$ is an equivalence because we take the
reflexive closure of some transitive and reflexive relation. But
as for observation based strategies, if we did not take the
reflexive closure, we could avoid
using the proposition $p_d$ and use the simpler formula 
\[\text{\tt AgreeOnLast'}: = \G \bigvee\limits_{a\in M}\R p_a\]

Indeed the relation would not be reflexive: in particular plays not ending in a dependence atom
would not be linked to any play, and $\R\phi$ would trivially
hold for any $\phi$ in these plays; this is why testing whether we are in a
dependence atom before enforcing that some $p_a$ must hold everywhere would no longer be needed.

\subsection{Dependence logic and games with imperfect information}

As we said, the evaluation game for Dependence Logic presented in the previous subsection  is
  said to be a game with imperfect information.  We do not agree,
  because the difference between games with perfect information and
  games with imperfect information (at least in the perfect recall
  setting, it is not as clear otherwise, see
  \cite{piccione1997absent}) lies in the fact that in the latter, some finite plays
  are related, in the sense that they are indistinguishable to one of
  the players, and this player must behave the same
  way in these related situations. Concerning the evaluation game for Dependence
  Logic, the difference with perfect-information games is that some plays are related, those ending in positions bound to
  the same atomic dependence formula $\dep(t_1,\ldots,t_n)$ with
  valuations agreeing on $t_1,\ldots,t_{n-1}$, and the valuations in
  these related positions must agree on  $t_n$.  So
  it is not that players should behave the same way in
  related situations, but rather that the players should
  \emph{have behaved} in such a way that the valuations for $t_n$ are
  the same in related situations.

  But it is true that there is a similarity between these two
  constraints on allowed strategies, as shown by looking at the
  formulas of the uniformity properties capturing observation-based strategies ({\tt SameAct}) and
  uniform strategies in the sense of Dependence Logic ({\tt
    AgreeOnLast}):
 \[\text{\tt
    SameAct}=\G(p_1\rightarrow\bigvee\limits_{a\in\Act}\!\!\!\R
  \fullmoon p_a) \mbox{\hspace{.3cm} and \hspace{.3cm}}\text{\tt
    AgreeOnLast} = \G(p_d\rightarrow \bigvee\limits_{a\in M}\R p_a)\]
  
 In the first case, the same thing must \emph{happen} in equivalent situations, whereas
  in the second case, the same thing must \emph{hold} in equivalent situations.

The resemblance is even more striking if we take the second versions:
 \[\text{\tt
    SameAct'}=\G\bigvee\limits_{a\in\Act}\!\!\!\R
   p_a \mbox{\hspace{.3cm} and \hspace{.3cm}}\text{\tt
    AgreeOnLast'} = \G\bigvee\limits_{a\in M}\R p_a\]

  Neither semantics games for Dependence Logic are games with imperfect
  information in the classical sense, nor games with imperfect
  information can be easily described using the uniform strategy
  notion of \cite{vaananen2007dependence}, but both can be characterized in a very
  similar way with our notion of uniform strategies.

\section{Regular relations} 
\label{sec-transducer}
The previous examples from the literature all fall into a particular
class of binary relations, so-called \emph{regular}\footnote{Actually,
  the genuine vocabulary is ``rational relations'', but we prefer to
  use ``regular relations'' to avoid any misleading terminology in the
  context of game theory and rational players.} relations. They are characterized by some
finite state machines called \emph{transducers}
\cite{berstel1979transductions}.
One way to see a transducer is to picture a nondeterministic automaton
with two tapes, an \emph{input} tape and an \emph{output} tape. The
transducer reads an input finite word on its input tape and writes out
a finite word on its output tape. Notice that this machine is in
general nondeterministic so that it may have several outputs for a
given input word. Hence, transducers  define binary relations. 

\begin{definition}
\label{def-trans}
A \emph{Finite State Transducer} (FST) is a tuple
$T=(Q,\Sigma,\Gamma,q_i,Q_F,\Delta)$, where $Q$ is a finite set of
states, $\Sigma$ is the \emph{input alphabet}, $\Gamma$ is the
\emph{output alphabet}, $q_i\in Q$ is a distinguished \emph{initial
  state}, $Q_F\subseteq Q$ is a set of \emph{accepting states}, and
$\Delta\subset Q\times (\Sigma \cup \{\epsilon\}) \times
(\Gamma \cup \{\epsilon\}) \times Q$ is a finite set of
\emph{transitions}.
\end{definition}

Intuitively, $(q,a,b,q')\in \Delta$ means that the transducer can move
from state $q$ to state $q'$ by reading $a$ and writing $b$ (both possibly $\epsilon$).
We also define the \emph{extended transition relation} $\Delta^*$,
which is the smallest relation such that:
\begin{itemize}
\item for all $q\in Q$, $(q,\epsilon,\epsilon,q)\in\Delta^*$, and
\item if $(q,w,w',q')\in \Delta^*$ and $(q',a,b,q'')\in \Delta$, then
  $(q,w\cdot a,w'\cdot b,q'')\in\Delta^*$.
\end{itemize}

In the following, for $q,q'\in Q$, $w\in \Sigma^*$ and $w'\in
\Gamma^*$, notation $q\labedgeext{w}{w'}q'$ means $(q,w,w',q')\in\Delta^*$.

\begin{definition}
\label{def-rel-trans}
  Let $T=(Q,\Sigma,\Gamma,q_i,Q_F,\Delta)$ be an FST. The \emph{relation
  recognized by} $T$ is \[[T]\egdef \{(w,w')\mid w\in\Sigma^*, w'\in
\Gamma^*, \exists q\in Q_F,\;q_i\labedgeext{w}{w'}q\}.\]
\end{definition}

In other words, a couple $(w,w')$ is in the relation recognized by $T$
if there is an accepting execution of $T$ that reads $w$ and outputs $w'$.

\begin{definition}
\label{def-rel-rat}
Let $\Sigma$ and $\Gamma$ be two alphabets. 
A binary relation $\leadsto\;\subseteq \Sigma^*\times \Gamma^*$ is \emph{regular} if it is recognized by an FST.
\end{definition}

We now recall some basic properties of regular relations that will be
useful later on. The first property regards intersection (regular
relations are not closed under intersection in general) and the second
property regards composition. The interested reader is referred to
\cite{berstel1979transductions} for technical details.

\begin{property}
\label{prop-rat}
Let $\Sigma$ and $\Gamma$ be two alphabets.
  Let $\leadsto\;\subseteq \Sigma^*\times\Gamma^*$ be a regular relation,
  and let $L$ and $L'$ be two regular languages over $\Sigma$ and
  $\Gamma$ respectively. Then $\leadsto \cap\; (L\times L')$ is also a
  regular relation.
\end{property}

Let $\Sigma,\Sigma',\Sigma''$ be three alphabets.
  Let $\leadsto_1\;\subseteq \Sigma^*\times\Sigma'^*$ and
  $\leadsto_2\;\subseteq \Sigma'^*\times\Sigma''^*$ be two binary
  relations.  

\begin{definition}
\label{def-comp}
The \emph{composition} of $\leadsto_1$ and $\leadsto_2$ is
$\leadsto_1\circ \leadsto_2\subseteq \Sigma^*\times \Sigma''^*$,
defined by:
\[\leadsto_1\circ\leadsto_2=\{(\rho,\rho'')\mid\exists
\rho'\in\Sigma'^*, \rho\leadsto_1\rho' \mbox{ and }\rho'\leadsto_2\rho''\}\]
\end{definition}

\begin{property}{\cite{elgot1965relations}}
\label{prop-comp}
If $\leadsto_1$ and $\leadsto_2$ are two regular
  relations recognized respectively by $T_1$ and $T_2$ , then
  $\leadsto_1\circ \leadsto_2$ is also regular, and the composition
  of the transducers
  $T_1\circ T_2$ recognizes $\leadsto_1\circ\leadsto_2$.
\end{property}


We close this section by taking two examples of binary
relations  involved in Section~\ref{sec-literature} and showing that
they are regular. In fact, one can check that they all are.

\begin{example} 
  \label{example1}
We first consider an example induced by the
  imperfect-information setting of Section~\ref{sec-K-based}. Let
  $\ga_{imp}=(V,E,v_I,\val,\sim)$ be an imperfect-information game arena as described in Section~\ref{sec-K-based}. 
  Relation $\approx\;\subseteq \PlaysFin ^2$ is an \emph{observational equivalence} over
  plays, generated
  by the equivalence $\sim$ between positions. Consider the FST $T_{obs}$
  depicted in Figure~\ref{fig-FSTexample-K-based}, with a unique
  initial state (ingoing arrow) that is also the final state (two
  concentric circles). It reads an
  input letter (a position) and outputs any position that is
  $\sim$-equivalent to it. This FST recognizes a relation $\simeq$ over
  $V^*\times V^*$, such that $\approx\;=\;\simeq\;\cap\; \PlaysFin^2$.
  Because  $\PlaysFin$
  is a regular language (one can see $\ga$ as a finite state
  automaton), Proposition~\ref{prop-rat} gives that
  $\simeq\cap\; \PlaysFin^2$ is also a regular relation, and in fact  
  $\ga\times T_{obs}\times \ga$ is a transducer that precisely recognizes $\approx$.\\
\end{example}
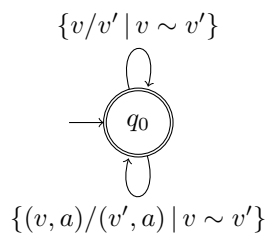
\begin{figure}[htp]
\begin{center}
\begin{tikzpicture}
  \node[state,initial,initial text=,accepting] (q_0) {$q_0$}; 
\path[->] (q_0) edge   [loop above] node {$\{v/v' \,|\,
    v \sim v'\}$} (q_0) 
(q_0) edge   [loop below] node {$\{(v,a)/(v',a)\,|\, v\sim v'\}$} (q_0);
\end{tikzpicture}
\end{center}
\caption{$T_{obs}$, an FST for the equivalence relation $\simeq$ of Sections~\ref{sec-K-based} and \ref{sec-opacity}.}
\label{fig-FSTexample-K-based}
\end{figure}

\newcommand{\oblue}{{\tt blue}}
\newcommand{\opink}{{\tt pink}}

\begin{example} 
\label{example2}
Consider another binary relation that is also an equivalence, but
induced by some alphabetic morphism $h : V \to {\mathcal O} \cup
\{\epsilon\}$: two plays $\rho$ and $\rho'$ are equivalent, written
$\rho\equiv_{\mathcal O}\rho'$, whenever $h(\rho)=h(\rho')$. This
example generalizes the one of Section~\ref{sec-diagnosis} where
the alphabetic morphism is a mere projection.  In order to draw an FST
for the relation $\equiv_{\mathcal O}$, we need to fix the set ${\mathcal O}$. Assume ${\mathcal O}$ has only two elements
$\oblue$ and $\opink$ so that any position in the game is either
observed as $h(v)=\oblue$ or $h(v)=\opink$ or unobserved (the case
$h(v)=\epsilon$). The FST $T_{\mathcal O}$ that recognizes $\equiv_{\mathcal O}$ is drawn in
Figure~\ref{fig-FSTexample-diagnosis}. 
 Once again, one should take the product of $T_{\mathcal O}$ with the
 game arena to restrict the relation to $\PlaysFin\times\PlaysFin$.
Remark that contrary to the
case of equivalence  $\approx$ in Figure~\ref{fig-FSTexample-K-based}, the
 equivalence $\equiv_{\mathcal O}$ does not preserve the
length of plays.
\end{example}

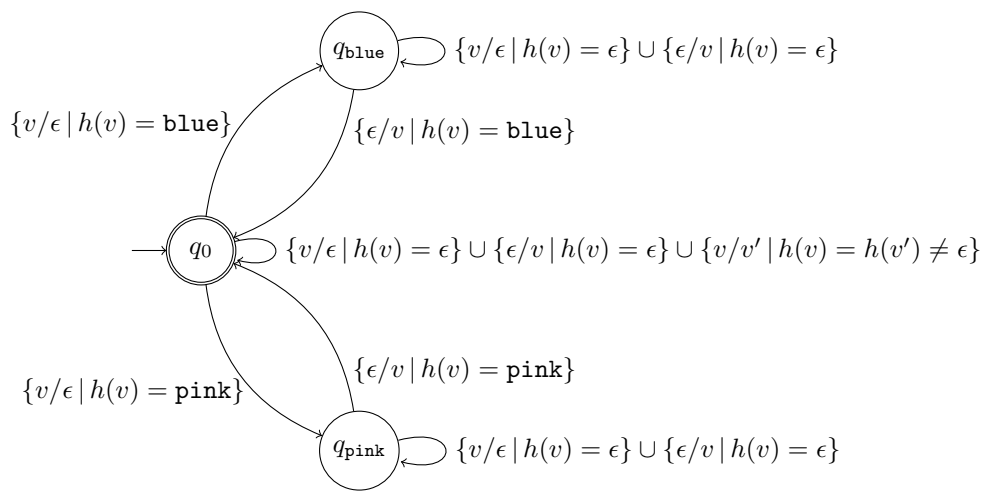
\begin{figure}[htp]
\begin{center}
\begin{tikzpicture}
  \node[state,initial,initial text=,accepting] (q_0) {$q_0$}; 
\node (bidon) [right=1.5cm of q_0] {}; 
\node[state] (qo1) [above=2cm of bidon] {$q_\oblue$};
  \node[state] (qo2) [below=2cm of bidon] {$q_\opink$}; 
\path[->] 

(q_0) edge   [loop right] node {$\{v/\epsilon \,|\,
    h(v)=\epsilon\} \cup \{\epsilon/v \,|\, h(v)=\epsilon\} \cup 
    \{v/v' \,|\, h(v)=h(v')\neq\epsilon\}$} (q_0) 

(q_0) edge [bend left] node [pos=0.5,left] 
{$\{v/\epsilon\,|\, h(v)=\oblue\}$} (qo1) 

(qo1) edge [loop right] node
  {$\{v/\epsilon\,|\, h(v)=\epsilon\}\cup  \{\epsilon/v\,|\,
    h(v)=\epsilon\}$} (qo1) 

(qo1) edge [bend left] node [pos=0.2,right] {$\{\epsilon/v\,|\, h(v)=\oblue\}$} (q_0)

(q_0) edge [bend right] node [pos=0.6,left] {$\{v/\epsilon\,|\, h(v)=\opink\}$} (qo2) 
 
(qo2) edge [loop right] node
  {$\{v/\epsilon\,|\, h(v)=\epsilon\}\cup  \{\epsilon/v\,|\,
    h(v)=\epsilon\}$} (qo2) 

(qo2) edge [bend right] node [pos=0.2,right] {$\{\epsilon/v\,|\,
    h(v)=\opink\}$} (q_0) ;

\end{tikzpicture}
\end{center}
\caption{$T_{\mathcal O}$, an FST for the equivalence relation
  $\equiv_{\mathcal O}$.}
\label{fig-FSTexample-diagnosis}
\end{figure}

Example~\ref{example1} is an example of a simple
transducer, but in fact all relations of Section~\ref{sec-literature}
can be recognized by the transducer given in
Example~\ref{example2}, for an appropriate alphabetic morphism. 
In the next section, we will only consider regular relations
over plays, i.e.\ the relation $\leadsto$ in the model is recognized by some FST. 

\section{Automated synthesis of fully-uniform strategies}
\label{sec-decision}
In this section, we study the problem of synthesizing fully-uniform
strategies.  We restricy finite arenas. Motivated by Section~\ref{sec-transducer}, we only
consider regular relations, and as a consequence we always assume that
the semantic relation between plays is described by a finite state
transducer $T$. Still, when it is clear from the context, we write
$\leadsto$ instead of $[T]$.  Also, because we only consider full
uniformity and no longer the strict one, it is understood in the
semantics of a formula that the universe $\Pi$ is the set of plays of
the considered game, hence we omit it when it is clear from the
context.  Furthermore, we will sometimes make the semantic relation
between plays explicit. All these conventions yield a notation of the
form $\pi,i\models_\leadsto \varphi$ meaning that
$\PlaysInf,\pi,i\models \varphi$ with the $\R$-modality semantics
based on the binary relation $\leadsto$. 

Finally, in this section, the size $|\ga|$ of an arena $\ga$ is
the number of positions, the size $|T|$ of a transducer $T$ is the
number of states,
and  $|\phi|$ is the size of formula $\phi$.
 
\begin{definition} For each $n\in\setn$, we define the decision problem $\fus{n}$ by:
$$
\begin{array}{ll}
  \fus{n}:=\{(\ga,T,\phi)\,|\, & \ga \text{ is an arena and } ([T],\phi) \text{ is a uniformity property\footnote{as in Definition~\ref{def-unifprop}} with } 
  \varphi \in \lang_n \text{ such that} \\
  & \text{ there exists a } ([T],\phi)\text{-fully-uniform strategy for Player $1$ in } \ga\}
\end{array}
$$
The \emph{fully-uniform strategy problem} is \FUS$:=\bigcup_{n\in\setn}\fus{n}$

\end{definition}

For an instance $(\ga,T,\phi)$ of \FUS, its size is defined as $|(\ga,T,\phi)|:=|\ga|+|T|+|\phi|$.  





\begin{theorem}
\label{theo-fusn}
$\fus{n}$ is in $2$-\EXPTIME\ for $n\leq 2$, and in $n$-\EXPTIME\ for $n>2$.
\end{theorem}

\begin{corollary}
The fully-uniform strategy problem is decidable.
\end{corollary}
The rest of this section is dedicated to the proof of
Theorem~\ref{theo-fusn}. We describe a powerset construction for a new
arena (Section~\ref{sec-powerset-arena}) and a way to lift the
semantic relation between plays to this powerset construction
(Section~\ref{sec-lift-trans}). Next we show how to exploit this
construction to reduce membership in $\fus{n+1}$ to membership in
$\fus{n}$ (Section~\ref{sec-elimination}).

\subsection{Powerset arena}
\label{sec-powerset-arena}


In games with imperfect information, the information set of a player
after a finite play is the set of positions that are
consistent with what she has observed.
We  define a similar notion in our setting,
and  we show that the regularity assumption on the
 relation is sufficient to compute information sets and build
a  powerset construction arena in which formulas of the kind $R\phi$
where $\phi\in \ltl$ can be
evaluated positionally.  


\begin{definition}
\label{def-infoset}
Let $\ga$ be an arena, $\leadsto\;\subseteq \PlaysFin^2$ and $\rho\in\PlaysFin$. The \emph{information
set}  after the finite play $\rho$ is the set of terminating positions
of related finite plays. Formally, 
$I(\rho)=\{v'\mid
\exists \rho'\cdot v'\in\PlaysFin,\;\rho\leadsto\rho'\cdot v'\}$.
\end{definition}



For an arena $\ga$ and a transducer $T$ over $\PlaysFin$,
we construct a powerset arena $\pow{\ga}$ in which formulas of the
form $\R\phi$ can be evaluated positionally  when $\phi\in\langn{0}$.

 Unlike classic powerset
 constructions \cite{reif84,chatterjee2006algorithms}, in our setting, the new information set
 after a move in the game
 cannot be computed knowing only the previous information set and the
 new position. 
 To compute information sets,
 we need to simulate the nondeterministic execution of  $T$,
 taking as an input the sequence of positions played and writing as
 output the related plays.  This is why in our construction, positions do not contain
 directly information sets, but rather  we add in the positions of the
 game sufficient information on the current configuration of the
 transducer.

 More precisely, two
 things are necessary: the set of states  the transducer may be in
 after reading the sequence of positions
 played so far, plus for each of these states the set of possible last positions
 written on the output tape (because of nondeterminacy, different
 executions can end up in configurations with same state but different last letters on
 the output tape). We only need to remember the last letter on the
 output tape, and not the whole tape, because the information set which we aim at computing is
 just the set of the \emph{last} positions of related plays.

 So positions are of the form
 $(v,\Set,\Last)$, where $\Set\subseteq Q$ is
 the set of possible current states of $T$,
 and $\Last:\Set\rightarrow \mathcal P(V)$ associates to a state $q\in \Set$ the
 set of the possible last positions on the
 output tape of $T$ if the current state is $q$. 
 The transitions in this arena follow the ones in $\ga$, we just
 have to maintain the additional information about the configuration of the
 transducer.
In order to define the initial position 
$\pow{v}_I=(v_I,\Set_I,\Last_I)$ of $\pow{\ga}$ we need to simulate the execution of $T$ starting from its initial state
and reading $v_I$. To do so, we introduce   an artificial position
$\pow{v}_{-1}$ that initializes the
transducer before reading the first position of a
play. Concretely, $\pow{v}_{-1}=(v_{-1},\Set_{-1},\Last_{-1})$, with
$v_{-1}\notin V$ a fresh position,
$\Set_{-1}=\{q_i\}$ because before starting the transducer is in its
initial state, and $\Last_{-1}(q_i)=\emptyset$ because nothing is
written on the output tape.  


\begin{definition}
\label{def-powerset}
Let $\ga=(V,E,v_I,\val)$ be an arena and 
$T=(Q,V,q_i,Q_F,\Delta)$ be an FST such that $[T]\subseteq \PlaysFin^2$. We define
the arena $\pow{\ga} =(\pow{V},\pow{E},\pow{v}_I,\pow{\val})$ by:
\begin{itemize}
\item $\pow{V}=V\times
  \mathcal P(Q) \times (Q \rightarrow \mathcal
P(V)) \uplus \{\pow{v}_{-1}\}$
\item $(u,\Set,\Last)\;\pow{\rightarrow}\;(v,\Set',\Last')$ if 
\begin{itemize}
\item $u=v_{-1}$ and $v=v_I$, or $u\rightarrow v$,
\item $\Set'=\{q'\mid\exists q\in \Set,\exists \lambda'\in
  V^*,\;q\labedgeext{v}{\lambda'}q'\}$ and
\item
    $\Last'(q')=\{v'\mid \exists q\in \Set, \exists \lambda' \in V^*,
     q\labedgeext{v}{\lambda'\cdot v'}q', \mbox{ or }
    q\labedgeext{v}{\epsilon}q' \mbox{ and }v'\in \Last(q)\}$ 
\end{itemize}
\item $\pow{v}_I$ is the only $\pow{v}\in\pow{V}$ such that 
  $\pow{v}_{-1}\;\pow{\rightarrow}\; \pow{v}$.
\item $\pow{\val}(\pow{v})=\val(v)$ if $\pow{v}=(v,\Set,\Last)$.
\end{itemize}
\end{definition}
Notice that $\pow{\ga}$ has size $|\pow{\ga}|=O(|\ga|\times 2^{|T|}\times 2^{|T||\ga|})$.\\


Regarding the definition of transitions, the first point means that
transitions in $\pow{E}$ follow those in $E$, except for the only
transition leaving $\pow{v}_{-1}$, that is used to define $\pow{v}_I$.
The second point for the definition of $S'$ expresses that when we move from $u$ to
$v$ in $\ga$, we give $v$ as an input to the transducer. So the set of
states the transducer can be in after this move
is exactly the set of states that can be reached from one of the
previous possible states by reading $v$ and writing some sequence of positions
(possibly $\epsilon$). We use the notation $\lambda$ because without
loss of generality one could assume that the transducer can only output
sequences of positions that form a valid path in the game graph. But
it is not important here, the assumption that $[T]\subseteq
\PlaysFin^2$ is sufficient. 
Finally, the third point for the definition of $\Last'$ captures that if some position $v'$ is at the end of the output
tape after the transducer read $v$ and reached $q'$, it is either
because while reading $v$ the last letter it wrote is $v'$, or it wrote
nothing and $v'$ was already the end of the output tape before reading $v$.  

To finish with, $\pow{v}_I$ is the only successor of $\pow{v}_{-1}$,
and the valuation of a position in the powerset construction is the
valuation of the underlying position in the original arena.

\begin{remark}
\label{rem-compute-finite}
We need to clarify the following definitions:
\begin{itemize}
\item  $\Set'=\{q'\mid\exists q\in \Set,\exists \lambda'\in
V^*,\;q\labedgeext{v}{\lambda'}q'\}$ and
\item $\Last'(q')=\{v'\mid \exists q\in \Set, \exists \lambda' \in V^*,
q\labedgeext{v}{\lambda'\cdot v'}q', \mbox{ or }
q\labedgeext{v}{\epsilon}q' \mbox{ and }v'\in \Last(q)\}$
\end{itemize}
Indeed, in each one, there
can be infinitely many such $\lambda'$ because of transitions
that read nothing on the input tape. Still, 
$\Set'$ and $\Last'$ can be computed
in linear time in the size of $\Delta$.
To do so, for each $q$ in $\Set$, we compute
$\Setaux{q,v}=\{(q',v')\mid \exists \lambda'\in V^*,
q\labedgeext{v}{\lambda'\cdot v'}q', \mbox{ or
}q\labedgeext{v}{\epsilon}q' \mbox{ and }v'\in \Last(q)\}$.
$\Set'$ and $\Last'$ can be easily reconstructed from $\cup_{q\in \Set}\Setaux{q,v}$. 
For $q\in \Set$, computing $\Setaux{q,v}$ can be done by depth-first
search, by first reading $v$ and then only $\epsilon$, and remembering
the last output. The search can be stopped when we reach a state that
has already been visited.
     
\end{remark}

Let us take an arena $\ga$ and an FST $T$ such that $[T]\subseteq
\PlaysFin^2$. For a position $\pow{v}$ of $\pow{\ga}$, we will access the
different components of the position with the notations
$\pow{v}.v,\pow{v}.\Set,\pow{v}.\Last$.

There is a  natural mapping  $f:\PlaysInf\rightarrow\pow{\PlaysInf}$:
for an infinite play $\pi\in\ga$, we define $f(\pi)$ as the only play
$\pow{\pi}$ in $\pow{\PlaysInf}$ such that $\pow{\pi}[0].v \cdot\pow{\pi}[1].v\cdot
\pow{\pi}[2].v \ldots = \pi$. This is well defined because from a
position $\pow{u}=(u,\Set,\Last)$ in $\pow{V}$, for  $v\in V$ such that
$u\rightarrow v$, there is a unique move $\pow{u}\;\pow{\rightarrow}\;\pow{v}$ such that
 $\pow{v}.v=v$. It is easy to see that $f$ is a bijection. From now on
 we will slightly abuse notations: for a play $\pi\in\PlaysInf$,
 $\pow{\pi}$ will denote $f(\pi)$, and for
 $\pow{\pi}\in\pow{\PlaysInf}$, $\pi$ will denote $f^{-1}(\pow{\pi})$.
 \emph{idem} for finite plays.


\begin{definition}
\label{def-localinfoset}
  Let $\pow{v}=(v,S,\Last)\in\pow{V}$ be a position in the powerset
  construction. Its \emph{local information set} $\pow{v}.I$ is
  defined by:
  \[\pow{v}.I:=\bigcup_{q\in S \cap Q_F}\Last(q)\] 
\end{definition}

Definition~\ref{def-localinfoset} means that a position $v'$ is in the information set
after a play $\pow{\rho}$ if and only if there is an execution of the
transducer on the word seen so far that terminates in an accepting
state with $v'$ the last output.  Actually, the
local information sets correspond to the real information
sets as expressed by the following proposition. 

\begin{proposition}
\label{prop-infoset}
For all $\pow{\rho} \in \pow{\PlaysFin}$, $\last(\pow{\rho}).I=I(\rho)$.
\end{proposition}

\noindent The rest of this subsection is dedicated to the proof of Proposition~\ref{prop-infoset}.

\begin{lemma}
\label{lem-last}
Let $\pow{\rho} \in \pow{\PlaysFin}$, and let
$\pow{v}:=\last(\pow{\rho})$. Then, $\pow{v}.\Set = \{q \mid \exists
\lambda'\in V^*, q_i\labedgeext{\rho}{\lambda'}q\}$, and for each $q
\in \pow{v}.\Set$, $\pow{v}.\Last(q) =\{v'\mid \exists \lambda' \in V^*,
q_i\labedgeext{\rho}{\lambda'\cdot v'}q\}$.
\end{lemma}

\begin{proof}

The proof is by induction on $\pow{\rho}$. 

\begin{description}
\item[Case $\pow{v}_I$.] We note
  $\pow{v}_I=(v_I,\Set_I,\Last_I)$, and start with the left-right
  inclusions for both equalities. Let
$q'\in \Set_I$ and $v'\in\Last_I(q)$.
 By definition of $\Last_I$ and $S_I$ there is a $q$ in
$\Set_{-1}=\{q_i\}$ (so $q=q_i$) and  a $\lambda'\in V^*$ such that
$q_i\labedgeext{v_I}{\lambda'} q'$, which proves the first inclusion.  Since
$\Last_{-1}(q_i)=\emptyset$ by definition, the case $\lambda'=\epsilon$ is not
possible. So there exists $\lambda''$ such
that $\lambda'=\lambda''\cdot v'$, which gives us the second inclusion.

The proofs for the two right-left inclusions are straightforward
applications of the definitions
of $\Set_I$ and $\Last_I$.

\item [Case $\pow{\rho}\cdot \pow{u} \cdot \pow{v}$,
  $\pow{\rho}\in\pow{\PlaysFin}\cup\{\epsilon\}$.] We note
  $\pow{u}=(u,\Set,\Last)$ and $\pow{v}=(v,\Set',\Last')$.  For
  the two left-right inclusions, let $q'\in \Set'$ and $v'\in
  \Last'(q')$. By definition of $\Set'$ and $\Last'$ there is a $q$ in
  $\Set$ and a $\lambda'_1\in V^*$ such that $q\labedgeext{v}{\lambda'_1} q'$,
  and either $\lambda'_1=\lambda'\cdot v'$ for some $\lambda'$, or
  $\lambda'_1=\epsilon$ and $v'\in\Last(q)$. 
  By induction hypothesis, we have that $S=\{q \mid \exists \lambda'\in
  V^*, q_i\labedgeext{\rho\cdot u}{\lambda'}q\}$,
  so there exists $\lambda'_2\in V^*$ such that
  $q_i\labedgeext{\rho\cdot u}{\lambda'_2} q$, and
  by transitivity,
  $q_i\labedgeext{\rho\cdot u\cdot
    v}{\lambda'_2\cdot\lambda'_1}q'$. This proves the first
  left-right inclusion. For the second one we split the two cases
  for $\lambda'_1$.
  \begin{itemize}
  \item If $\lambda'_1=\lambda'_3\cdot v'$ for some $\lambda'_3$, then by transitivity we have $q_i\labedgeext{\rho\cdot u\cdot
    v}{\lambda'_2\cdot\lambda'_3\cdot v'}q'$, which proves the second left-right inclusion. 
    \item If $\lambda'_1=\epsilon$, then $v'\in \Last(q)$. By
  induction hypothesis there is some $\lambda'_3$ such that
  $q_i\labedgeext{\rho \cdot u}{\lambda'_3\cdot
    v'}q$. By transitivity we obtain   $q_i\labedgeext{\rho\cdot
    u\cdot v}{\lambda'_3\cdot
    v'}q'$, which also  proves the second left-right inclusion. 
  \end{itemize}

  Now for the first right-left inclusion, take $q'$ and $\lambda'$ such
  that
  $q_i\labedgeext{\rho\cdot u\cdot v}{\lambda'}q'$. Necessarily
  there exist $\lambda'_1,\lambda'_2$ and $q$ such that
  $q_i\labedgeext{\rho\cdot u}{\lambda'_1}q$,
  $q\labedgeext{v}{\lambda'_2}q'$ and $\lambda'_1\cdot\lambda'_2=\lambda'$. By
  induction hypothesis $q\in \Set$, so by definition of $\Set'$,
  $q'\in\Set'$. For the second right-left inclusion,  take
  $q'\in \Set'$, and take $v'$ and $\lambda'$ such that
  $q_i\labedgeext{\rho\cdot u\cdot v}{\lambda'\cdot
  v'}q'$. Again, necessarily   there exist $\lambda'_1,\lambda'_2$ and $q$ such that
  $q_i\labedgeext{\rho\cdot u}{\lambda'_1}q$,
  $q\labedgeext{v}{\lambda'_2}q'$ and $\lambda'_1\cdot\lambda'_2=\lambda'\cdot v'$. By
  induction hypothesis $q\in \Set$. We distinguish two cases. 
  \begin{itemize}
    \item If $\lambda'_2=\epsilon$, then
  $\lambda'_1=\lambda' \cdot v'$, hence
  $q_i\labedgeext{\rho \cdot u}{\lambda'\cdot
    v'}q$. By induction hypothesis, $v'\in\Last(q)$, so by definition
  of $\Last'$, because $q\in \Set$ and $q\labedgeext{v}{\epsilon}q'$, we obtain
  $v'\in \Last'(q')$.  
  \item If $\lambda'_2=\lambda'_3\cdot v'$ for some $\lambda'_3$, then
by definition of $\Last'$, because $q\in \Set$, we have $v'\in
\Last'(q')$.   
\end{itemize}
\end{description}
This finishes the induction.\end{proof}

We can now terminate the proof of Proposition~\ref{prop-infoset}: Let
$\pow{\rho} \in \pow{\PlaysFin}$, and let $\pow{v}=\last(\pow{\rho})$
be of the form $\pow{v}=(v,\Set,\Last)$.  We remind that (Definition~\ref{def-infoset}) $I(\rho)=\{v'\in V\mid
\exists \rho'\cdot v'\in\PlaysFin,\; \rho\;\leadsto \rho'\cdot v'\}$.\\

We start with the left-right inclusion. Let $v'\in \pow{v}.I$. By definition,
$\pow{v}.I=\bigcup_{q\in S\cap Q_F} \Last(q)$,
so $v'\in \Last(q)$ for some $q\in S\cap Q_F$. By
Lemma~\ref{lem-last}, there exists $\lambda'\in V^*$ such that
$q_i\labedgeext{\rho}{\lambda'\cdot v'} q$, and
because $q\in Q_F$, we have that
$(\rho, \lambda'\cdot v')\in [T]$, which
implies that $\rho \leadsto \lambda'\cdot v'$. Since $\leadsto\;\subseteq
\PlaysFin^2$, 
 $\lambda'\cdot v'\in \PlaysFin$, hence $v'\in
I(\rho)$.

For the right-left inclusion, take $v'\in
I(\rho)$. There exists $\rho'$ such that
$\rho'\cdot v'\in \PlaysFin$ and
$\rho \leadsto \rho'\cdot v'$. By definition
of $T$, there exists $q\in Q_F$ such that
$q_i\labedgeext{\rho}{\rho'\cdot v'}q$.
By Lemma~\ref{lem-last}, $q\in S$, and $v'\in \Last(q)$. Since $q\in
S\cap Q_F$, $v'\in \pow{v}.I$.

\subsection{Lifting transducers}
\label{sec-lift-trans}

Let $\ga$ be an arena, $T$ an FST such that $[T]\subseteq
\PlaysFin^2$, and let $\pow{\ga}=\power{\ga}{T}$.  We describe how to
build a transducer $\pow{T}$ that lifts $[T] \subseteq \PlaysFin
\times \PlaysFin$ to $\pow{\PlaysFin}\times\pow{\PlaysFin}$.

We note $\Tdown$ for the deterministic transducer that computes $f$,
the bijection that maps a
play $\pow{\rho}\in\pow{\PlaysFin}$ to the underlying play
$\rho\in\PlaysFin$, and $\Tup$ for the deterministic transducer
that computes $f^{-1}$. Both are easily built from $\pow{\ga}$, and 
$|\Tdown|=|\Tup|=O(|\pow{\ga}|)$.

\begin{definition}
\label{def-lift-trans}
The \emph{lift} of transducer $T$ is $\pow{T}=\Tdown \circ T \circ \Tup$.
\end{definition}

\noindent Notice that $|\pow{T}|=O(|\pow{\ga}|\times|T|\times |\pow{\ga}|)$. \\

In the following we let $\leadsto\; = [T]$ and $\pow{\leadsto}\;=[\pow{T}]$.
The following  proposition follows directly from the definitions of $\pow{\ga}$ and $\pow{T}$:



\begin{proposition}
\label{prop-pow-equiv}
For every $\phi\in\lang$, $\pi\in\PlaysInf$,
$i\geq 0$,
\[\pi,i\models_\leadsto \phi \mbox{ iff
}\pow{\pi},i\models_{\pow{\leadsto}} \phi\]

\end{proposition}

\subsection{$\R$-elimination}
\label{sec-elimination}


We establish that given an instance of
$\fus{n+1}$, we can build in exponential space and time an equivalent instance of $\fus{n}$. 

\begin{proposition}
\label{prop-decr}

For all instance $(\ga,T,\phi)$ of $\fus{n+1}$, there exists an
instance $(\ga',T',\phi')$ of $\fus{n}$ computable in time exponential in
$|(\ga,T,\phi)|$ such that:

\begin{itemize}
\item $(\ga,T,\phi)\in\fus{n+1}$ iff $(\ga',T',\phi')\in\fus{n}$
\item $|\ga'|=O(2^{(|\ga|+|T|)^2})$ 
\item $|T'|= O(2^{O(|\ga|+|T|)^2})$
\item $|\phi'|=O(|\phi|)$
\end{itemize}

\end{proposition}

\noindent The rest of the section is dedicated to the proof of Proposition~\ref{prop-decr}. Let $(\ga,T,\phi)$ be an instance of $\fus{n+1}$.


\begin{lemma}
  \label{lem-R-info}
  Let $\pi\in\PlaysInf$, $i\geq 0$, and $\phi$ be an
  $\ltl-$formula. 
 \[\pi,i\models_\leadsto \R\phi \mbox{ iff }
 \ga,u\models \phi \mbox{ for all }u\in I(\pi[0,i]).\]
\end{lemma}

\begin{proof}
  We start with the left-right implication. Suppose that
  $\pi,i\models \R\phi$ holds, and take $u\in I(\pi[0,i])$.
  We need to prove that $\ga,u\models\phi$. To do so, we take
  $\pi'\in\TracesInf(u)$  an infinite trace starting in $u$ and we
  prove that $\pi',0\models\phi$.
  Since $u\in I(\pi[0,i])$, by definition of the information set, there exists
  $\rho\cdot u\in\PlaysFin$ such that $\pi[0,i]\leadsto\rho\cdot u$. We
  let $j=|\rho|$ and $\pi''=\rho\cdot\pi'$ be such that $\pi''[j]=u$.
  Clearly, $\pi'' \in \PlaysInf$, and $\pi[0,i]\leadsto\pi[0,j]$.
  Since $\pi,i\models \R\phi$ holds,
  we have that $\pi'',j\models \phi$. And because
  $\phi\in\ltl$, the fact that it holds at some point of a trace only
  depends on the future of this point, hence
  $\pi''[j,\infty],0\models \phi$, \ie $\pi',0\models \phi$.

For the right-left implication, suppose that $ \ga,u\models \phi
\mbox{ for all }u\in I(\pi[0,i])$, and take $\pi'\in\PlaysInf$,
$j\geq 0$ such that $\pi[0,i]\leadsto\pi'[0,j]$. We have that $\pi'[j]\in
I(\pi[0,i])$, so $\ga,\pi'[j]\models \phi$. Because $\pi'[j,\infty]$
is in $\TracesInf(\pi'[j])$, we have that $\pi'[j,\infty],0\models
\phi$, hence $\pi',j\models \phi$.
\end{proof}

\begin{lemma}
  \label{lem-position}
  Let $\pow{\pi}\in\pow{\PlaysInf}$, $i\geq 0$, and let $\phi$ be an
  $\ltl-$formula. 

 \[\pow{\pi},i\models_{\pow{\leadsto}} \R\phi \mbox{ iff }
 \ga,u\models \phi \mbox{ for all }u\in \pow{\pi}[i].I.\]

\end{lemma}

\begin{proof} Let $\pow{\pi}\in\pow{\PlaysInf}$ and $i\geq 0$. By Proposition~\ref{prop-pow-equiv}, for any $\phi \in LTL$, 
  $\pow{\pi},i\models_{\pow{\leadsto}} \R\phi \mbox{ iff
  }\pi,i\models_\leadsto \R\phi$, and by Lemma~\ref{lem-R-info},
  $\pi,i\models_\leadsto \R\phi \mbox{ iff } \ga,u\models \phi \mbox{
    for all }u\in I(\pi[0,i])$. Now, by Proposition~\ref{prop-infoset},
  $\pow{\pi}[i].I=I(\pi[0,i])$, which concludes the proof.
\end{proof}

We now define how formulas of the kind $\R\phi$ can be replaced by new
atomic propositions, and how positions of the powerset arena can be
marked with these new propositions: 
To an arena $\ga$, a formula $\phi\in\lang$ and a subformula
$\R\psi\in\langn{1}\cap\subf(\phi)$ (if any), we associate a fresh
atomic proposition $p_{\R\psi}$ that occurs neither in $\ga$ nor in
$\phi$.

\begin{definition}
\label{def-newphi}
For $\phi\in\langn{n+1}$, we define $\pow{\phi}:=\phi[p_{\R\psi}/\R\psi \mid \R\psi\in\langn{1}\cap
  \subf(\phi)]$.
\end{definition}

\begin{example}
$\pow{\R \fullmoon \R q}=\R \fullmoon p_{\R q}$ 
\end{example}


\begin{definition}
\label{def-decr}
For an instance $(\ga,T,\phi)$ of $\fus{n+1}$, we define the instance
$\pow{(\ga,T,\phi)}$ of $\fus{n}$ as $(\ga',T',\phi')$ by:
\begin{itemize}
\item If $\pow{\ga}=(\pow{V},\pow{E},\pow{v_I},\pow{\val})$, then
  $\ga'=(\pow{V},\pow{E},\pow{v_I},\pow{\val}')$, with

{\centering
  
$\pow{\val}'(\pow{v})=\pow{\val}(\pow{v})\cup \{p_{\R\psi}\mid
\R\psi\in\langn{1}\cap\subf(\phi)\mbox{ and }\forall u\in\pow{v}.I, \ga,u\models
\psi\}.$

}

\item $T'=\pow{T}$
\item $\phi'=\pow{\phi}$
\end{itemize}
\end{definition}  

From now on, for an instance $(\ga,T,\phi)$ of $\fus{n+1}$, we abuse
notation by writing $\pow{\ga}=(\pow{V},\pow{E},\pow{v_I},\pow{\val})$
for the \emph{modified} powerset construction of
Definition~\ref{def-decr}.

\begin{lemma}
  \label{lemma-decr}
  Take an instance $(\ga,T,\phi)$ of $\fus{n+1}$, and let
  $(\pow{\ga},\pow{T},\pow{\phi})=\pow{(\ga,T,\phi)}$.
  Then for all $\pi\in\PlaysInf$ and $i\geq 0$,
  \[\pi,i\models_\leadsto \phi
  \mbox{ iff }\pow{\pi},i\models_{\pow{\leadsto}} \pow{\phi} \]
\end{lemma}

\begin{proof}
  By Proposition~\ref{prop-pow-equiv}, $\pi,i\models_\leadsto \phi
  \mbox{ iff }\pow{\pi},i\models_{\pow{\leadsto}} \phi$, so it only
  remains to show that $\pow{\pi},i\models_{\pow{\leadsto}} \phi$ iff
  $\pow{\pi},i\models_{\pow{\leadsto}} \pow{\phi}$. We prove it by
  induction on $\phi$. The cases $\phi=p, \phi=\neg \psi, \phi=\psi\vee\psi',
  \phi=\fullmoon\psi, \phi=\psi\until\psi'$ are trivial. 
  It remains to consider the case $\phi=\R\psi$, which
  decomposes into two subcases depending on $\depth(\psi)$:
    \begin{itemize}
    \item If $\depth(\psi)>0$, then
    $\pow{\phi}=\R\pow{\psi}$.
    We then have:
    \begin{align*}
      \pow{\pi},i\models_{\pow{\leadsto}} \R\psi &\mbox{ iff }\forall
      \pow{\pi}',j \mbox{
        s.t. }\pow{\pi}[0,i]\;\pow{\leadsto}\;\pow{\pi}'[0,j],\; \pow{\pi}',j\models_{\pow{\leadsto}}\psi
      \\
      &\mbox{ iff }\forall
      \pow{\pi}',j \mbox{
        s.t. }\pow{\pi}[0,i]\;\pow{\leadsto}\;\pow{\pi}'[0,j],\;\pow{\pi}',j\models_{\pow{\leadsto}}\pow{\psi}
      \mbox{\hspace{2cm}(by induction hypothesis)}\\
      &\mbox{ iff }\pow{\pi},i\models_{\pow{\leadsto}} \R\pow{\psi}
    \end{align*}
    
    \item If $\depth(\psi)=0$, that is $\psi\in LTL$, then $\pow{\phi}=p_{\R\psi}$.
      \begin{align*}\pow{\pi},i\models_{\pow{\leadsto}} \R\psi
      &\mbox{ iff }  \ga,u\models \psi \mbox{ for all }u\in\pow{\pi}[i].I
      &\mbox{\hspace{2cm}(by Lemma~\ref{lem-position})}\\
      &\mbox{ iff } p_{\R\psi}\in\pow{\val}(\pow{\pi}[i])
      &\mbox{\hspace{2cm}(by Definition~\ref{def-decr})}\\
      &\mbox{ iff } \pow{\pi},i\models_{\pow{\leadsto}} p_{\R\psi} &
    \end{align*}
  \end{itemize}
\end{proof}

We can now achieve the proof of Proposition~\ref{prop-decr}.  Take an
instance $(\ga,T,\phi)$ of $\fus{n+1}$. We show that
$(\pow{\ga},\pow{T},\pow{\phi})=\pow{(\ga,T,\phi)}$ is a good candidate.  Notice that the natural
bijection between $\PlaysFin$ and $\pow{\PlaysFin}$ induces a
bijection between strategies $\sigma$ in $\ga$ and strategies
$\pow{\sigma}$ in $\pow{\ga}$, such that for every strategy $\sigma$
in $\ga$, if we note $\pow{\out(\sigma)}:=\{\pow{\pi}\mid
\pi\in\out(\sigma)\}$, then $\out(\pow{\sigma})=\pow{\out(\sigma)}$.

Let $\sigma$ be $([T],\phi)$-fully-uniform in $\ga$. If $\pow{\pi}\in\out(\pow{\sigma})$, then
$\pow{\pi}\in\pow{\out(\sigma)}$, hence $\pi\in\out(\sigma)$. Because
$\sigma$ is $([T],\phi)$-fully-uniform,
$\pi,0\models_\leadsto\phi$. By Lemma~\ref{lemma-decr} we conclude
that $\pow{\pi},0\models_{\pow{\leadsto}}\pow{\phi}$, which means that 
$\pow{\sigma}$ is $([\pow{T}],\pow{\phi})$-fully-uniform in
$\pow{\ga}$. Since $\depth(\pow{\phi})=\depth(\phi)-1=n$, $\pow{(\ga,T,\phi)} \in \fus{n}$.

Assume $\pow{(\ga,T,\phi)} \in \fus{n}$, that is there exists a
strategy $\pow{\sigma}$ which is
$([\pow{T}],\pow{\phi})$-fully-uniform in $\pow{\ga}$. Any play $\pi
\in \out(\sigma)$ is uniquely associated to a play $\pow{\pi}\in
\pow{\out(\sigma)}=\out(\pow{\sigma})$, which by assumption satisfies
$\pow{\pi},0\models_{\pow{\leadsto}}\pow{\phi}$. By
Lemma~\ref{lemma-decr}, $\pi,0\models_{\leadsto}\phi$, which shows
that $\sigma$ is $([T],\phi)$-fully-uniform in $\ga$.

This achieves the proof of the first point of
Proposition~\ref{prop-decr}.  For the second point, recall
Definition~\ref{def-powerset} that gives $|\pow{\ga}|=O(|\ga|\times
2^{|T|}\times 2^{|T||\ga|})$, so that
$|\pow{\ga}|=O(2^{(|\ga|+|T|)^2})$. The third point is derived from
this second point and
Definition~\ref{def-lift-trans}: $|\pow{T}| =O(|\pow{\ga}|\times|T|\times |\pow{\ga}|) =O( 2^{(|\ga|+|T|)^2} \times |T|\times 2^{(|\ga|+|T|)^2} )  =O(2^{O(|\ga|+|T|)^2})$. 
Finally, the fourth point stating that $|\pow{\phi}|=O(|\phi|)$ is immediate by Definition~\ref{def-newphi}. \\

It remains to prove that $\pow{(\ga,T,\phi)}$ is computed in time
exponential in $|(\ga,T,\phi)|$. Clearly, the powerset construction
(Section~\ref{sec-powerset-arena}) and the lifting of the transducer
(Section~\ref{sec-lift-trans})  both take exponential time in
$|\ga|+|T|$, hence in $|(\ga,T,\phi)|$. The marking phase in the
$\R$-elimination (Definition~\ref{def-decr}) involves model checking at most $|\phi|$ $LTL$-formulas
on each position of the original arena $\ga$. Model checking an $LTL$-formula in a given position requires polynomial space in
$|\ga|+|\phi|$ \cite{sistla85}. Since \PSPACE\ $\subseteq$ \EXPTIME, it is exponential in
time.  All in all, we need to model check an $LTL$-formula at most $|\ga|\times |\phi|$ times, so
the whole marking phase is done in time exponential in
$|(\ga,T,\phi)|$. We conclude that
$\pow{(\ga,T,\phi)}$ can be computed in time exponential in
$|(\ga,T,\phi)|$, which terminates the proof of Proposition~\ref{prop-decr}.



\subsection{Complexity of $\fus{n}$}

In this subsection we describe an algorithm that decides whether
an instance $(\ga,T,\phi)$ is in $\FUS$, and we establish upper bounds
for the $\fus{n}$ problem, for each $n \in setn$.

Algorithm~\ref{fig-algo-FUS} describes our decision procedure. It takes
as an entry an instance $(\ga,T,\phi)$ of $\FUS$, and returns \texttt{true} if it is
a positive instance\footnote{i.e., there exists a $([T],\phi)$-fully-uniform strategy in
   $\ga$.}, \texttt{false} otherwise. To do so, starting from
$(\ga_0,T_0,\phi_0)=(\ga,T,\phi)$, it successively applies the
construction described in Subsection~\ref{sec-elimination} to
eliminate $\R$ operators in $\phi$ and to ultimately reduce the problem to solving an
equivalent $LTL$ game.

It is known that solving $LTL$ games has a time complexity
doubly-exponential in the size of the formula, and that it is actually
\2EXPTIME-complete \cite{pnueli89b}.  We remind that solving an $LTL$
game $(\ga,\phi)$, in the automata-theoretic formulation of this
problem \cite{vardi1991verification}, can be done by the following
procedure, that we will call in $LTLGameSolver$ (see
\cite{pnueli89b,alur2004deterministic}).  First, compute a
nondeterministic Büchi tree automaton that accepts trees whose
branches all verify the formula. This automaton is of size exponential
in $|\phi|$. Then, by for example Safra's construction \cite{safra88},
build an equivalent deterministic Rabin automaton $\mathcal A_\phi$
with a number of states doubly-exponential in $\phi$, and a number of
pairs exponential in $\phi$.  Then, with a linear cost in $|\ga|$,
transform the arena $\ga$ into a nondeterministic tree automaton
$\mathcal A_\ga$ that accepts all strategies of Player 1 in
$\ga$. Then, there exists a strategy whose outcomes all satisfy $\phi$
if and only if the product Rabin automaton $\mathcal A_\phi\times
\mathcal A_\ga$ accepts some tree. Deciding the emptiness of a Rabin
tree automaton can be done time $O((\ell m)^{3m})$, where $\ell$ is
the number of states and $m$ is the number of pairs of the Rabin
automaton \cite{rosner1991modular}. Provided that for the product
Rabin automaton $\mathcal A_\phi\times \mathcal A_\ga$ we have
$\ell=|\ga|\times 2^{2^{|\phi|}}$ and $m=2^{|\phi|}$, we finally obtain
the following upper bound:

\begin{proposition}
\label{prop-upbound-ltl}
Solving an $LTL$ game $(\ga,\phi)$ takes time
  $|\ga|^{2^{O(|\phi|)}}$.
\end{proposition}

  It is
  important to keep the size of the arena and the size of the formula
  apart for the moment, instead of just saying that it is
  doubly-exponential in the size of the entry, because in our decision procedure, the size of the
  iterated powerset constructions suffers a exponential blow-up, contrary to the successive formulas
  whose sizes remains unchanged (and even decrease since
  subformulas are replaced with atomic propositions).

\begin{algorithm}[h]
  \SetAlgoLined
  \KwIn{$(\ga,T,\phi)$}
  \KwOut{\texttt{true} if $(\ga,T,\phi)\in\FUS$, \texttt{false}
    otherwise}
  $(\ga_0,T_0,\phi_0):=(\ga,T,\phi)$\;
  $k:=0$\;
  \While{$\depth(\phi_k)>0$}{
    $(\ga_{k+1},T_{k+1},\phi_{k+1}):= \pow{(\ga_k,T_k,\phi_k)}$\;
    $k:=k+1$\;
  }
  \KwRet{LTLGameSolver$(\ga_k,\phi_k)$}\vspace{3pt}
\caption{Decision procedure for the problem \FUS.}
\label{fig-algo-FUS}

\end{algorithm}

Theorem~\ref{theo-fusn} as announced at the beginning of
Section~\ref{sec-decision} can now be proved.

Let $n\in\mathbb N$, and let $(\ga,T,\phi)$ be an instance of
$\fus{n}$.  If $n=0$, the body of the {\bf while} instruction is not
executed, and we immediately call $LTLGameSolver(\ga_0,\phi_0)$. By
Proposition~\ref{prop-upbound-ltl}, this call takes time
$|\ga_0|^{2^{O(|\phi_0|)}}$, hence it is $2$-\EXPTIME\ in
$|(\ga,T,\phi)|$. We next answer the case $n>0$, and will distinguish the two
particular cases $n=1$ and $n=2$.  

For convenience, we introduce notations for iterated exponential
functions: for $k,n\in\mathbb N$,
$\itexp{k}{n}=2^{2^{\cdots^{2^{n}}}}\big\}\small{k}$.

 



\begin{lemma}
\label{lem-Gi}
For every $0\leq k\leq n$, $|\ga_k|=|T_k|=\itexp{k}{O(|\ga|+|T|)^2}$.
\end{lemma}

\begin{proof}
By induction on $k$, using Proposition~\ref{prop-decr}.
\end{proof}

If $(\ga,T,\phi)$ is an instance of $\fus{n}$, for $1\leq k\leq n$, by
Proposition~\ref{prop-decr}, the execution of the $k$-th loop takes
time exponential in $|(\ga_{k-1},T_{k-1},\phi_{k-1})|$. Hence by
Lemma~\ref{lem-Gi}, the time complexity for the $k$-th loop is
$\itexp{1}{\itexp{k-1}{O(|\ga|+|T|)^2}+|\phi|}=2^{|\phi|}\itexp{k}{O(|\ga|+|T|)^2}$,
and the execution of the whole {\bf while}
instruction takes time $\theta_{\mbox{\sc while}}$ where: 
\begin{align*}
\theta_{\mbox{\sc while}} 
& =\sum_{k=1}^n
2^{|\phi|}\itexp{k}{O(|\ga|+|T|)^2} \\
 & =2^{|\phi|}\itexp{n}{O(|\ga|+|T|)^2}
\end{align*}  

By Proposition~\ref{prop-upbound-ltl} and Lemma~\ref{lem-Gi}, solving the final $LTL$ game
$(\ga_n,\phi_n)$ takes time   $\theta_{\mbox{\sc
    ltl}}$, where:
\begin{align*}
  \theta_{\mbox{\sc ltl}}
  &=|\ga_n|^{2^{O(|\phi_n|)}}\\
  &=\itexp{1}{\itexp{n-1}{O(|\ga|+|T|)^2}*2^{O(|\phi|)}}
\end{align*}

We obtain that Algorithm~1 runs in time
$\theta=\theta_{\mbox{\sc while}}+\theta_{\mbox{\sc ltl}}$, but since (for
$n>0$), $\theta_{\mbox{\sc while}}$ is negligible compared to
$\theta_{\mbox{\sc ltl}}$, we obtain:

\begin{equation}
\label{eq-theta}
\theta =\itexp{1}{\itexp{n-1}{O(|\ga|+|T|)^2}*2^{O(|\phi|)}}
\end{equation}
Additionally, for $n=1$, the double exponential complexity stems from
the size of the formula. For $n=2$, because the size of the arena has
taken two exponentials, the double exponential complexity comes both
from the size of the formula and the size of the arena. Afterwards,
since the arena keeps growing exponentially while the size of the
formula remains the same, the complexity comes essentially from the
size of the arena. This achieves the proof of Theorem~\ref{theo-fusn}.

Note that the subroutine $LTLGameSolver(\ga_n,\phi_n)$ of Algorithm~1,
based on the automata-theoretic procedure of \cite{pnueli89b}, does not
merely decide the existence of a winning strategy, but actually builds
one (if any). Recall also that forgetful strategies are sufficient for
$LTL$ games, as they are particular cases of regular games which enjoy the
``Forgetful Determinacy'' \cite{zeitman1994unforgettable}. By the
natural bijection invoked in the proof of Proposition~\ref{prop-decr}
between strategies in a powerset arena and strategies in the original
arena, one can trace the strategy in $\ga_n$ back to a
$([T_0],\phi_0)$-fully-uniform strategy in the original game $\ga_0$.

\begin{corollary}
  Forgetful strategies are sufficient for full-uniformity properties.
\end{corollary}

\section{Discussion}
\label{sec-discussion}

We have investigated the concept of uniform strategies in two-player
turn-based infinite-duration games, motivated by the many instances
from the literature: games with imperfect information, games with
epistemic condition, non-interference, diagnosis and prognosis, and
Dependence Logic. Uniformity is addressed in the context of a semantic
binary relation between plays of the arena, which can arise from any
 reason to relate plays with each others, \eg an epistemic
feature.

In order to embrace all the examples we have encountered, and likely
many potential others, we were led to designing a formal language whose
sentences express the very uniformity properties of
strategies. Clearly, the language-based approach offers intuitive
definitions, while the set-theoretic one, which may capture a larger
class of uniformity properties, is much less readable. The particular
uniformity properties that have been addressed in the literature so
far (Section~\ref{sec-literature}) can now more easily be compared.
Our language is an enrichment of the Linear-time Temporal logic $LTL$
\cite{gabbay80}, hence it is interpreted over plays. The additional
feature is captured by the modality $\R$ which quantifies universally
over related plays. Whether this quantification ranges over all
plays in the arena or just over outcomes of the considered strategy
yields two variants of uniform strategies, namely fully-uniform and
strictly-uniform strategies.

The general procedure to decide the fully-uniform strategy problem is
non-elementary. This may be the price to pay for a generic solution
for arbitrarily complex uniformity properties, and we conjecture that the
fully-uniform strategy problem is non-elementary hard. However, bounding the
$\R$-depth of the formulas gives an elementary bound complexity, which
seems incidentally to be the case for all the examples of
Section~\ref{sec-literature}: only formulas whose $\R$-depth is one
are needed, so that the generic procedure has ``only'' a double
exponential time complexity. Notice that \cite{maubert2011opacity}
obtained a tighter (optimal) single exponential time bound for solving
games with opacity condition, which corresponds to the fixed formula
$\G \neg \R p_S$ of $\R$-depth equal to one and to a simple fixed
binary relation between plays (see Section~\ref{sec-opacity}). Notice
that if we fix a formula of $\R$-depth $1$, the time complexity
Equation~\eqref{eq-theta} of our procedure
collapses to a single exponential time complexity in the size of the arena and of the transducer.\\

Our results can be extended and commented in many respects. We give here some of them. \\

First, the choice we have made to rely on an enrichment of the $LTL$ logic can
be questioned -- although this logic regarding properties of time is
acknowledged in many respects. We may try to extend the synthesis
procedure to a much richer logic like the Linear-time $\mu$-calculus,
a language extending standard linear time temporal logic with fix-point
operators. But the current procedure relies on a bottom-up traversal
of the parse-tree of the formula $\phi$, which cannot be generalized
to formulas with arbitrary fix-points. The $LTL$ logic falls into
the very particular so-called \emph{alternation-free} fragment of the
Linear-time $\mu$-calculus, where fix-points do not
interplay. Significant progress in understanding this extended setting
need being pursued.\\

Second, we considered a single semantic binary relation between
plays. One may wonder whether the case of several relations
$\leadsto_i$, yielding modalities $\R_i$ at the language level, can be
investigated at the algorithmic level. We foresee a generalization of
our powerset construction by synchronizing the execution of all the
transducers of the relations. We however remain cautious regarding the
success of this approach since closely related topics such as the
\emph{decentralized} diagnosis problem is known to be undecidable
\cite{sengupta2002decentralized}. Still the question is important as
it would unify our setting with the Epistemic Temporal Logic $ETL$ of
\cite{halpern1989complexity} and bring light on the automated
verification of $ETL$ definable properties for \emph{open systems} (see the
\emph{module}-checking problem of \cite{kupferman97a}). \\

Additional comments are needed to fully understand the contribution,
in particular regarding the recent developments of alternating-time
epistemic logic \cite{van2003cooperation,jamroga2004agents,dima2010model}. The two settings are close but
incomparable. With the uniform strategy concept, we aim at extending
the range of (qualitative) properties of strategies by means of binary
relations between plays, and at exploiting those properties to
synthesize particular strategies. Instead, alternating-time epistemic
logics offer a way to quantify over strategies that achieve $ETL$-like
properties, hence they are not synthesis-oriented, and moreover, they
do not handle arbitrary relations between plays. Unifying the two
settings is a real challenge; we would need to design a (necessarily
more complex) language that incorporates the specification of the
relation(s) between plays.

\bibliographystyle{alpha}
\bibliography{games,logic,opacity,misc,strategy,classic,automata}


\end{document}